\documentclass[tbtags,authorcolumns,numberwithinsect]{no-lipics}

\usepackage{amssymb}
\usepackage{mathtools}
\usepackage{bm}
\usepackage{stmaryrd}
\usepackage[draft]{fixme}
\usepackage{booktabs}
\usepackage{tikz-3dplot}
\usepackage{caption}
\usepackage{subcaption}
\usepackage{cases}
\usepackage{colortbl}

\usetikzlibrary{shapes, positioning}

\SetKwComment{Comment}{$\triangleright$\ \rm}{}

\usetikzlibrary{arrows,arrows.meta,shapes.misc, decorations.pathreplacing,decorations.pathmorphing,calligraphy, patterns.meta}
\usetikzlibrary{matrix, positioning, fit, shapes, calc, backgrounds, shapes.geometric, backgrounds}
\tikzset{dotmark/.style={circle,fill,inner sep=1.5pt}}
\tikzset{emptymark/.style={circle,draw,fill=white,inner sep=1.5pt}}
\tikzset{crossmark/.style={thick,inner sep=1.5pt}}

\newcommand{\eps}{\varepsilon}
\newcommand{\A}{\mathcal{A}}

\SetKwInput{KwInput}{Input}
\SetKwInput{KwOutput}{Output}

\def\ShowAuthNotes{0}
\ifnum\ShowAuthNotes=1
\newcommand{\authnote}[3]{\textcolor{#3}{[{\bf #1:} { {#2}}]}}
\else
\newcommand{\authnote}[3]{}
\fi

\newcommand{\Oh}{\mathcal{O}}
\newcommand{\Ohtilde}{\tilde{\Oh}}

\DeclareMathOperator*{\poly}{poly}

\DeclareMathOperator*{\argmin}{arg\,min}
\DeclareMathOperator*{\argmax}{arg\,max}

\def\fragmentco#1#2{\bm{[}\,#1\,\bm{.\,.}\,#2\,\bm{)}}
\def\fragmentoc#1#2{\bm{(}\,#1\,\bm{.\,.}\,#2\,\bm{]}}
\def\fragmentoo#1#2{\bm{(}\,#1\,\bm{.\,.}\,#2\,\bm{)}}
\def\fragment#1#2{\bm{[}\,#1\,\bm{.\,.}\,#2\,\bm{]}}

\newcommand{\ceil}[1]{\lceil #1 \rceil}

\newcommand{\AT}{\textsf{AT}\xspace}

\newcommand{\mC}{\mathcal{C}}
\newcommand{\mH}{\mathcal{H}}

\newcommand{\mP}{\mathcal{P}}

\newcommand{\bF}{\mathbf{F}}

\newcommand{\bG}{\mathbf{G}}
\newcommand{\bH}{\mathbf{H}}

\newcommand{\bT}{\mathbf{T}}

\renewcommand{\epsilon}{\varepsilon}

\usetikzlibrary{calc,positioning,backgrounds,fit}

\DeclarePairedDelimiter\abs{\lvert}{\rvert}
\DeclarePairedDelimiter\norm{\lVert}{\rVert}

\makeatletter
\let\oldabs\abs
\def\abs{\@ifstar{\oldabs}{\oldabs*}}
\let\oldnorm\norm
\def\norm{\@ifstar{\oldnorm}{\oldnorm*}}
\makeatother

\def\problembox#1{%
    \vspace{2mm}%
    \noindent\fbox{%
    \begin{minipage}{.985\linewidth}%
        #1
    \end{minipage}%
    }%
    \vspace{2mm}%
}

\newcommand{\defdynamicproblem}[4]{%
    \problembox{%
        \textbf{#1}\\
        {\bf{Input:}} #2 ~\\
        {\bf{Update:}} #3 ~\\
        {\bf{Query:}} #4
    }%
}

\makeatletter
\renewenvironment{cases}{%
  \matrix@check\cases\env@cases
}{%
  \endarray\right.%
}
\def\env@cases{%
  \let\@ifnextchar\new@ifnextchar
  \left\lbrace
  \def\arraystretch{1.1}%
  \array{@{\;}c@{\quad}l@{}}%
}
\makeatother

\def\mid{\ensuremath :}

\def\emptyset{\varnothing}

\newcommand\thefont{\expandafter\string\the\font}

\makeatletter
\AtBeginDocument{\let\core@label\label}
\newcommand{\symlabel}[2]{%
  #1% Print the symbol
  \ifmeasuring@\else
    \bgroup
    \def\@currentlabel{\ensuremath{#1}}% Set the label to your symbol
    \core@label{#2}% Apply the pure label, completely bypassing amsmath's hijack
    \egroup
  \fi
}
\makeatother

\title{Asteroidal Triples and Friends\\ in (Likely Necessary) Matrix Multiplication Time}

\author{Jakob Nogler}{Massachusetts Institute of Technology\\Cambridge, United States}{jnogler@mit.edu}{https://orcid.org/0009-0002-7028-2595}{Supported by the Akamai Presidential Fellowship and by a fellowship from the Siebel Scholars Foundation}
\author{Agn\`es Totschnig}{Massachusetts Institute of Technology\\Cambridge, United States}{agnest@mit.edu}{https://orcid.org/0009-0004-7716-8183}{Supported by a MathWorks Fellowship}
\author{Virginia Vassilevska Williams}{Massachusetts Institute of Technology\\Cambridge, United States}{virgi@mit.edu}{https://orcid.org/0000-0003-4844-2863}{Supported by NSF Grant CCF-2330048, BSF Grant 2024233 and a Simons Investigator Award}

\authorrunning{J. Nogler, A. Totschnig, and V. Vassilevska Williams}
\titlerunning{Asteroidal Triples and Friends in (Likely Necessary) Matrix Multiplication Time}

\nolinenumbers

\begin{document}
\pagenumbering{roman}
\maketitle
\begin{abstract}
In a seminal paper, Kratsch and Spinrad [SODA 2003 and SIAM Journal on Computing 2006] studied several important problems in structural graph theory that anticipated two trends in fine-grained complexity, both from the upper- and lower-bound side: \begin{itemize}
    \item From the upper-bound side, for each of these problems the authors give an \emph{intermediate} running time strictly \emph{between $n^{\omega}$ and $n^3$}, where $\omega<2.371177$ is the matrix multiplication exponent. When optimized, this running time can be made $\Ohtilde(n^{(3+\omega)/2})$. By now, many more problems can be solved in this running time:
    directed unweighted or node-weighted APSP, a variety of matrix products such as max-min and dominance, RNA Folding and many more. For some of these, the community has also developed tools suggesting that this intermediate running time is necessary.

    \item From the lower-bound side, the authors give lower bounds of $n^{\omega-o(1)}$, conditioned on three natural hardness hypotheses on the complexity of quantified variants of triangle detection that account for different types of hardness of such problems.
\end{itemize}

Despite this, it has remained open for over 20 years which of these two pieces of evidence reflects the truth for the problems that Kratsch and Spinrad study: whether the intermediate running time is indeed necessary and stronger lower bounds hold, or whether the hardness evidence of Kratsch and Spinrad is essentially tight. In this paper, we resolve this question and show that the latter is the case: all problems admit an $\Ohtilde(n^{\omega})$ time algorithm.

As an additional contribution, we make partial algorithmic progress on further problems in structural graph theory of intermediate complexity, improving over previous work [Chiu and Lu, STACS 2022 and Information and Computation 2024][Spinrad, Discrete Applied Math 2004].

\end{abstract}

% \clearpage
% \thispagestyle{plain}
% \tableofcontents
% \clearpage
\pagenumbering{arabic}

% {\bf TODO: intro, abstract}

\section{Introduction}
Many problems in graph and matrix algorithms have natural cubic running times in terms of $n$, the number of nodes in the graph or the number of rows/columns of the matrices. For some perhaps infamous examples, such as All-Pairs Shortest Paths (APSP) in weighted graphs, this cubic running time has persisted for many decades, except for $n^{o(1)}$ improvements, leading to conjectures that cubic time is the best possible (e.g.,\cite{vsurvey,focsy,focsyj}). For other problems, substantial polynomial improvements have been obtained. The prime example for this is matrix multiplication.
The exponent $\omega$ is defined as the smallest real value for which $n\times n$ matrices can be multiplied in $\Oh(n^{\omega+\eps})$ time for all $\eps>0$. It has been known
 since 1969 that $\omega<3$ \cite{strassen}, and while the popular belief among computer scientists is that $\omega=2$, the best known upper bound is $\omega<2.371177$ \cite{DupontEtAl26}.

Via reductions to matrix multiplication, a huge variety of problems have truly subcubic\footnote{$O(n^{3-\eps})$ time for constant $\eps>0$.} time algorithms: many linear algebraic operations (matrix inverse, linear system solving, LUP decompositions and more), a variety of ``funny'' matrix products over other algebraic structures such as the max-min product of matrices \cite{VassilevskaWY09}, many graph problems in unweighted graphs such as unweighted APSP \cite{Zwick02,Seidel95}, structure similarity problems such as RNA-Folding \cite{BGSV19}, and many more. 

There are simple reductions from
Boolean matrix multiplication to many of these problems, showing
that $n^{\omega-o(1)}$ time seems {\em necessary} to solve them. This is true for instance for unweighted APSP and most ``funny'' matrix products. For even more problems, $n^{\omega-o(1)}$ time lower bounds can be proven under stronger but still plausible assumptions
such as that detecting a triangle in an $n$-node graph requires $n^{\omega-o(1)}$ time. For problems with $n^{\omega-o(1)}$ time conditional lower bounds and subcubic time upper bounds using matrix multiplication, a natural question is:

\begin{center}{\em Can these problems be solved in $\Ohtilde(n^\omega)$ time or is there a higher conditional lower bound?}\end{center}

The term ``intermediate problem'' has emerged to capture problems that are in truly subcubic time but whose running times are {\em strictly larger} than $n^\omega$.
Several recent works study such problems within the context of fine-grained complexity \cite{BKPR19, LPVW20, VWX20b, CVW21, F26}, providing conditional lower bounds $n^{c-o(1)}$ for $c>\omega$ for some of these problems. 

Most studied intermediate problems have running times that are $\Ohtilde(n^{2.5})$ if $\omega=2$; a popular running time is $\Ohtilde(n^{(3+\omega)/2})$.
This is the best known running time for the Dominance Product \cite{M91}, max-min matrix product and All-Pairs Bottleneck Paths \cite{VWY07, DP09}, All-Pairs Non-Decreasing Paths \cite{VW10, DGZ18, CJW19}, Monotone Min-Plus Product \cite{BGSV19, VWX20, GPVWX21, CDXZ22}, Language Edit Distance (also known as Scored Parsing \cite{aho1972minimum, s17, BGSV19}), RNA Folding \cite{NJ80, akutsu1999RNA, zakov2011RNA, venkatachalam2014RNA, s17}, Dyck Edit Distance \cite{s14,BGSV19,das2021improved,fried2024improved}, Tree Edit Distance \cite{M21, D23, NPSVWXY25}, Node-weighted APSP \cite{C10, Y09, AFJVWX25}, and Min-Equality Product \cite{VWX20b}.

One of the very first papers to prove $n^{\omega-o(1)}$ lower bounds under plausible conjectures for a variety of graph problems, while also providing subcubic time algorithms using matrix multiplication, is a seminal SODA 2003 paper by Kratsch and Spinrad \cite{KS06}.
Almost all of the problems studied by \cite{KS06}  were solved in
$\Oh(n^{2.79})$ time. With a small change\footnote{The neighborhood connected components algorithm of \cite{KS06} uses an inefficient way to list witnesses of Boolean matrix multiplication. Just by
replacing Lemma 2.1 in \cite{KS06} with the approach of \cite{AlonGMN92,GasieniecKL09} to list up to $f$ witnesses per pair in $\Ohtilde(f\mathsf{MM}(n,n/f,n))$ time, one obtains a running time $\Ohtilde(f\mathsf{MM}(n,n/f,n)+\mathsf{MM}(n,n^2/f,n))$ for neighborhood-deleted components which gives $\Ohtilde(n^{(3+\omega)/2})$ via blocking or $\Oh(n^{2.616})$ with rectangular matrix multiplication. For the rest of the problems, in \cite{KS06} there is an additional $\Ohtilde(n^{(3+\omega)/2})$ cost after the neighborhood-deleted components are computed.} the algorithms can be made to run in
 $\Ohtilde(n^{(3+\omega)/2})$ time. This motivates the following questions.

\begin{center}
\emph{Are these problems intermediate, or can they be solved in $\Ohtilde(n^\omega)$ time? \\ How do these problems fit into the recently established fine-grained hierarchy?}
\end{center}

Even the name of Kratsch and Spinrad's paper ``Between $\Oh(mn)$ and $\Oh(n^\alpha)$'' (here $\alpha=\omega$) suggests that the problems in the paper are intermediate. Despite much recent progress on the fine-grained complexity of intermediate problems, the problems studied in Kratsch and Spinrad's paper have not seen improvements since the paper was written more than 20 years ago.

The main contribution of our work is to show that these problems are not intermediate. They, in fact, admit conditionally tight $\Ohtilde(n^\omega)$ time algorithms.

\subsection{The problems studied by Kratsch and Spinrad}
Kratsch and Spinrad \cite{KS06} investigate the following algorithmic problems in structural graph theory that have very strong motivation, e.g., in the study of perfect graphs, \AT-free graphs and triangulated graphs (see Appendix~\ref{app:motivation} for details). Given an (undirected) graph $\bG = (V,E)$, the problems are:

\begin{itemize}
    \item Computing the \emph{neighborhood-deleted components} of $\bG$: Specifically, for each $v \in V$, compute the connected components of $\bG[V \setminus N[v]]$ (the graph obtained by removing $v$ and its neighborhood).
    \item Listing or finding the \emph{extremities} of $\bG$: These are vertices $v \in V$ such that $\bG[V \setminus N[v]]$ is connected. This can be derived directly from the previous problem.
    \item Determining whether $\bG$ has a \emph{star-cutset} or a \emph{clique-cutset}: A \emph{cutset} is a set of nodes $S \subseteq V$ such that $\bG[V \setminus S]$ is disconnected; it is a \emph{star-cutset} if there is a central node in $S$ connected to all other nodes in $S$, and a \emph{clique-cutset} if $\bG[S]$ is a clique.
    \item Recognizing \emph{\AT-free graphs}: This involves verifying whether $\bG$ contains an \emph{asteroidal triple}---an independent set of three vertices $x, y, z \in V$ such that, between any two vertices in the triple, there exists a path avoiding the neighborhood of the third vertex.
    \item Listing or finding \emph{dominating pairs}: These are distinct nodes $u, v \in V$ such that for every $u$-$v$ path $\pi$ and every node $w \in V$, the path $\pi$ contains at least one neighbor of $w$ or contains $w$ itself.
    (In other words, the vertex set of any path between these two vertices is a dominating set of the graph.)
    \item Listing or finding \emph{two-pairs}: These are distinct nodes $u, v \in V$ such that all induced (or equivalently, chordless) paths between $u$ and $v$ have length two.
\end{itemize}

For all of these problems, the authors present algorithms with running times strictly \emph{between $\Oh(n^3)$ and $\Oh(n^{\omega})$}. Specifically, as mentioned earlier,
an $\Oh(n^{2.79})$-time algorithm is provided for all problems, except for recognizing \AT-free graphs, which is solved in $\Oh(n^{4\omega/(\omega+1)}) = \Oh(n^{2.81})$ time.
As mentioned earlier, the $\Oh(n^{2.79})$ time can be replaced by $\Oh(n^{2.69})$ by slightly modifying the algorithms. 
Prior to \cite{KS06}, only $\Oh(nm)$-time algorithms were known for some of these problems, which still take $\Omega(n^3)$ time on dense graphs.

\subsection{Why $n^\omega$ seems needed for all the problems}
Spinrad and Kratsch provide evidence that solving the problems listed above requires at least $n^{\omega-o(1)}$ time. To support this, they identify three problems for which no algorithms faster than $\Oh(n^{\omega})$ are currently known (see \cref{fig:hardnessks} for the hardness proven by \cite{KS06} from these problems). For each of the three problems one hypothesizes that in the word-RAM model with $O(\log n)$ bit words, the problem requires $n^{\omega-o(1)}$ time.

\begin{description}
    \item[Determining whether a graph contains a triangle:] The corresponding Triangle Detection Hypothesis\label{hypothesis:a} that the problem requires $n^{\omega-o(1)}$ time has been repeatedly used within fine-grained complexity (e.g., \cite{vsurvey,focsyj,AbboudW14}). It is also supported by the subcubic equivalence \cite{focsy,focsyj} between triangle detection and Boolean matrix multiplication.

    \item[Determining whether a graph contains a dominated vertex:] A \emph{dominated vertex} is a node $x$ such that $N(x) \subseteq N[y]$ for some $y\neq x$. This problem turns out to be exactly equivalent to the Orthogonal Vectors (OV) problem in $n$ dimensions (see Section \ref{sec:dom} in the appendix for a proof). Whether this high-dimensional OV problem requires fast matrix multiplication was first asked\footnote{They actually asked whether distinguishing between diameter 2 and 3 in an undirected graph in subcubic time without matrix multiplication, but this problem is easily seen to be equivalent to OV in $n$ dimensions, e.g., via a reduction as in \cite{vdiam,vsurvey}.} almost 3 decades ago by Aingworth, Chekuri, Indyk and Motwani \cite{aingw}. The hypothesis
    \label{hypothesis:b} that OV in $n$ dimensions requires $n^{\omega-o(1)}$ time has been used in fine-grained complexity in recent years (see e.g., \cite{DalirrooyfardK21}).

    \item[Determining whether a graph contains a simplicial vertex:] A node $v$ is \emph{simplicial }if $\bG[N(v)]$ is a clique.
    The hypothesis that Simplicial Vertex requires $n^{\omega-o(1)}$ time
    \label{hypothesis:c} has also been used in recent years for hardness results in fine-grained complexity (see e.g., \cite{DJVW22}).
\end{description}

OV and Triangle detection are central problems in fine-grained complexity. Simplicial Vertex has been used only a handful of times as a source of hardness. Nevertheless, it is a very natural problem. In fact, the three hard problems that Kratsch and Spinrad propose are just different quantifications of the following form: 
\begin{align}
    Q_1 x\in V, Q_2 y\in N(x), Q_3 z\in N(x): (y,z)\in E. \label{eq:quant}
\end{align}
Triangle detection has quantifiers $Q_1Q_2Q_3=\exists\exists\exists$ in \cref{eq:quant}, Simplicial Vertex has $Q_1Q_2Q_3=\exists\forall\forall$, and OV in $n$ dimensions (and dominated vertex) is equivalent to the case $Q_1Q_2Q_3=\exists\exists\forall$ (this requires a bit of a proof since we need $y\in N(x)$ -- see Section \ref{sec:dom} in the appendix). Kratsch and Spinrad \cite{KS06} provide evidence that tight reductions between these different quantifications might be difficult due to an inherent difference in the non-deterministic complexity of the problems. Nevertheless, we do know that {\em fine-grained subcubic} (but non-tight w.r.t. $n^\omega$) reductions exist from Simplicial Vertex and Dominated Vertex to Triangle Detection because both these problems can be solved using Boolean Matrix Multiplication (BMM) and BMM is fine-grained equivalent to Triangle Detection \cite{focsy,focsyj}.

\begin{figure}[t]
        \centering
        \begin{tikzpicture}[
    font=\small\rmfamily, % Sets font size to small and uses the Roman (serif) family
    box/.style={
        rectangle, 
        rounded corners=4pt, 
        draw=black!70, 
        thick, 
        text width=15cm, 
        align=left, 
        inner sep=4pt
    }
]

% 1. Triangle-Free Hardness
\node[box, fill=blue!10] (triangle) at (0,0) {
    \textbf{\textcolor{blue!80!black}{Hard as Triangle-Free Recognition:}} Recognizing \AT-free graphs, finding a dominating pair, verifying a specific vertex pair is a dominating pair, identifying all extremities.
};

% 2. Dominated Vertex Hardness
\node[box, fill=green!10, below=0.2cm of triangle] (dominated) {
    \textbf{\textcolor{green!60!black}{Hard as Testing for a Dominated Vertex/High Dimensional OV:}} Determining whether a graph has a two-pair, determining whether a graph has a star cutset.
};

% 3. Simplicial Vertex Hardness
\node[box, fill=orange!10, below=0.2cm of dominated] (simplicial) {
    \textbf{\textcolor{orange!80!black}{Hard as Testing for a Simplicial Vertex:}} Determining whether a graph has a clique cutset.
};

\end{tikzpicture}
        \caption{Hardness of the problems solved in \cite{KS06}.}
        \label{fig:hardnessks}
\end{figure}

\subsection*{Our Results}

We show that all the problems presented by Spinrad and Kratsch can be solved in $\Ohtilde(n^{\omega})$ time. \footnote{To be precise, we show this for all problems in the journal version of \cite{KS06}. The conference version contains one additional problem (finding clique cutsets) for which an $\Oh(n^{\omega})$-time algorithm was found concurrently \cite{HTV05}. Later, this result was extended to cliqueset decompositions in matrix multiplication time \cite{CD18}.} Therefore, none of the presented problems fit into the aforementioned category of intermediate problems. If one is willing to believe that $n$-dimensional OV, Simplicial Vertex, and Triangle Detection all require $n^{\omega-o(1)}$ time, then our algorithms are also optimal.

We start with neighborhood-deleted components, noting that this is equivalent to a more general problem that we call the \emph{induced subgraph connected component} problem: Given an $n$-vertex graph $\bG = (V, E)$ and a family of subsets $\mathcal{S}$ of $V$, compute the connected components of $\bG[S]$ for each $S \in \mathcal{S}$. This equivalence can easily be seen to hold for $|\mathcal{S}| = \Theta(n)$: in one direction, we choose the sets to be the graphs obtained by deleting the (closed) neighborhoods; in the other direction, we introduce a clique containing one vertex for each $S \in \mathcal{S}$, where each such vertex's neighbors are exactly the vertices not in its corresponding set $S$.

We show that the induced subgraph connected component problem is \emph{equivalent} to multiplying an $|\mathcal{S}| \times n$ Boolean matrix by an $n \times n$ Boolean matrix. \footnote{We note that one (perhaps minor) reason we present this result for the induced subgraph connected component problem, as opposed to neighborhood deleted components, is that this equivalence with rectangular Boolean matrix product holds even when $|\mathcal{S}| = \omega(n)$.}
This equivalence demonstrates that we can compute the connected components of multiple graphs more efficiently, provided they are all induced subgraphs of a single underlying graph.

\begin{restatable*}{mtheorem}{ndcctobmm} \label{lem:ndcc_to_bmm}
    There is a deterministic algorithm that solves the induced subgraph connected component problem in time $\Ohtilde(\mathsf{MM}(|\mathcal{S}|, n, n))$ for an $n$-node input graph and a non-empty input set $\mathcal{S}$.
\end{restatable*}

Here $\mathsf{MM}(k, m, n)$ denotes the running time of multiplying a $k\times m$ Boolean matrix by an $m\times n$ Boolean matrix.

\begin{restatable*}{lemma}{ndcclb} \label{lem:ndcc_lb}
    Given two Boolean matrices $A$ and $B$ of size $s \times n$ and $n \times n$, one can construct in time $\Oh(n^2 + ns)$ an instance of the induced subgraph connected component problem on a $2n$-node graph with a set of size $s = |\mathcal{S}|$, such that solving this instance yields the product $A B$.
\end{restatable*}

As a corollary to \cref{lem:ndcc_to_bmm}, we obtain algorithms for the following problems (for details on how the preceding algorithm implies these results, we refer the reader to \cite[Corollaries 3.2, 3.3, and 3.4]{KS06}).

\begin{corollary}
    The following problems have a deterministic $\Ohtilde(n^{\omega})$-time algorithm on an $n$-node input graph: computing neighborhood-deleted components in a graph, listing all extremities in a graph and determining whether the graph has a star cut-set, determining whether two fixed vertices $x,y$ are a two-pair. \lipicsEnd
\end{corollary}

We also obtain an $\Ohtilde(n^{\omega})$-time algorithm for recognizing \AT-free graphs (a related result was previously known only for claw-free \AT-free graphs in matrix multiplication time \cite{HK02}), as well as for listing all dominating pairs and two-pairs.

\begin{restatable*}{mtheorem}{algatfree} \label{them:alg_atfree}
    There is a deterministic $\Ohtilde(n^{\omega})$-time algorithm that determines if a graph is \AT-free.
\end{restatable*}

Previously, the fastest time algorithm for recognizing \AT-free graphs ran in $\Oh(n^{2.82})$ time \cite{KS06}.

\begin{restatable*}{mtheorem}{dompairs} \label{thm:dompairs}
    There is a deterministic $\Ohtilde(n^{\omega})$-time algorithm for listing all dominating pairs.
\end{restatable*}

This improves upon Spinrad and Kratsch's $\Oh(n^{2.79})$ running time.

\begin{restatable*}{mtheorem}{twopairs} \label{thm:twopairs}
    There is a deterministic $\Ohtilde(n^{\omega})$-time algorithm for listing all two-pairs.
\end{restatable*}
The previous best running time was Spinrad and Kratsch's $\Oh(n^{2.79})$.

We remark that all our algorithms are deterministic, and the underlying algorithms are remarkably simple. Although the algorithms themselves are straightforward, the proofs establishing their correctness are occasionally more complex (e.g., recognition of \AT-free graphs); in such cases, the more involved proofs reveal new structural insights into the problems.

Finally, we study two related problems not addressed in the work of \cite{KS06}. While we only achieve partial progress on these, we suspect their time complexity may strictly exceed $\Oh(n^{\omega})$. Their fine-grained complexity offers an interesting direction for future study, so we leave both problems open to the community.

\subparagraph*{Induced Non-shortest Paths.}
The first problem we propose is a generalization of verifying whether two vertices form a \emph{two-pair}. More generally, given two vertices $u, v$ at distance $d(u,v) = k$, one can ask whether there exists an induced (equivalently, chordless) path between them that is not a shortest path (note that every shortest path is necessarily induced). Following previous work, we call such a path a \emph{$uv$-trail}. In this terminology, the problem of listing two-pairs is the all-pairs version of this problem restricted to vertex pairs with $k = 2$.

This problem is not new: it was recently introduced by Berger, Seymour, and Spirkl \cite{BSP21}, who gave an $\Oh(n^{18})$-time algorithm for a single fixed pair. This was very recently improved by Chiu and Lu \cite{CL24} to $\Oh(n^{2\omega})$, i.e., the time required to multiply two $n^2 \times n^2$ matrices. No lower bound was provided, nor any evidence that matrix multiplication is necessary. We show that matrix multiplication seems necessary, even for a single pair with $k = 3$, by giving a reduction from triangle detection.

\begin{restatable*}{lemma}{lbuvtrail} \label{lem:lb_uvtrail}
    For any $k \geq 3$, suppose there exists an algorithm that, in time $\mathsf{T} = \Omega(n^2)$, determines whether an $n$-node graph $\bG$ contains a $uv$-trail for two vertices $u, v$ with $d(u, v) = k$. Then there is an algorithm solving triangle detection in time $\Oh(\mathsf{T})$.
\end{restatable*}

We find this problem interesting from a fine-grained complexity viewpoint because it appears closely related to another area whose fine-grained landscape still has many open problems: the detection of induced patterns (see \cite{FKLL12, VWWWWY15, BKS18, DVVW21, DVW22}), where also runtimes of $\Oh(n^{2\omega})$ are common, e.g., for detecting a six-clique or other patterns.

Finally, we mention that algorithms for detecting $uv$-trails have been used to improve the runtime for recognizing graphs in which all holes have the same length. This was first achieved by Cook et al. \cite{CHPRSSTV24} in time $\Oh(n^{21})$ and subsequently improved by Chiu and Lu \cite{CL24} to $\Oh(n^{2\omega + 3})$.

\subparagraph*{Recognition of Quasi-Triangulated Graphs and Decremental Simplicial Vertex.}
The second problem is the recognition of another generalization of triangulated graphs, called \emph{quasi-triangulated} graphs. A graph $\bG$ is \emph{quasi-triangulated} if and only if every induced subgraph of $\bG$ contains a simplicial vertex or a \emph{cosimplicial vertex} (i.e., a simplicial vertex in the complement). Just as triangulated graphs can be recognized by iteratively removing a simplicial vertex from the resulting graph, quasi-triangulated graphs can be recognized similarly by iteratively removing a simplicial or cosimplicial vertex. Thus both classes can be recognized using the following subroutine.

\defdynamicproblem{Decremental Simplicial Vertex}{An unweighted graph $\bG = (V, E)$.}{Delete a given vertex $v$ together with all its incident edges.}{Does the current graph $\bG$ contain a simplicial vertex? If so, return one.}

Hence, to recognize whether a graph $\bG$ is quasi-triangulated, it suffices to maintain a decremental simplicial vertex subroutine for both $\bG$ and its complement $\bar{\bG}$, proceeding iteratively as follows: 
In each iteration, check whether $\bG$ or $\bar{\bG}$ contains a simplicial vertex. If such a vertex $v$ exists in either graph, remove $v$ from both $\bG$ and $\bar{\bG}$. If this elimination process continues until the graph is empty, $\bG$ is quasi-triangulated. Otherwise, if the process halts because neither $\bG$ nor $\bar{\bG}$ contains a simplicial vertex, the algorithm terminates and rejects $\bG$.

In reality, triangulated graphs can be recognized faster, in linear time \cite{RTS76}, using another simple algorithm. However, the fastest algorithm for quasi-triangulated graphs runs in time $\Oh(\min(n^{2.77}, nm))$ \cite{S04, H03}, where the recognition algorithm relies on the decremental simplicial vertex subroutine. We find this last problem especially interesting from a fine-grained point of view because edge-decremental triangle detection has recently been shown to admit a $\Ohtilde(n^{\omega})$-time algorithm \cite{BBFKS26}. In this light, it would be interesting to see whether Triangle Detection, Simplicial Vertex, or High-dimensional OV all retain a $\Ohtilde(n^{\omega})$-time algorithm, or whether some interesting difference emerges.

Our partial progress here is an improvement of the algorithm for Decremental Simplicial Vertex to $\Oh(n^{2.5275})$, which on sparse graphs can be balanced to be something slightly better.

\begin{restatable*}{theorem}{decrementalsimplicial} \label{thm:decrementalsimplicial}
    There is an algorithm for Decremental Simplicial Vertex that on an $n$-node graph runs in time
    \[
        \Oh\left( \ \textstyle\min_{\Delta \in \fragment{1}{n}} \left\{ n/\Delta \cdot \mathsf{MM}(n, \Delta, n) + m \cdot \Delta \right\} \ \right).
    \]
    In particular, with the current bounds on rectangular matrix multiplication\footnote{We used the bounds from \cite{ADVWXXZ25} plugged into van den Brand's optimization tool \cite{balancer}.} this gives a running time of $\Oh(n^{2.5275})$ on dense graphs. For sparse graphs, bounding the rectangular product naively by dividing it into $\Delta \times \Delta$ blocks yields a running time of
    $\Oh\left( m^{(3-\omega) / (4-\omega)} \, n^{3 / (4-\omega)} \right)$.
\end{restatable*}

Notably, the running time we obtain is the same as the best known running times for both directed unweighted APSP \cite{Zwick02} and Min Witness Product \cite{CKL07}.

\section{Overview}
\label{sec:to}

In the following section, we give an overview of the proofs of \cref{lem:ndcc_to_bmm}, \cref{them:alg_atfree}, and \cref{thm:dompairs}: the first because its proof is simple and can be explained briefly, and the second and third because they let us highlight the additional structural properties we establish for these problems.

\subsection{Induced Subgraph Connected Components in Matrix Multiplication Time}

Given a graph $\bG = (V,E)$ and a family of subsets $\mathcal{S} \subseteq 2^{V}$, our goal is to compute, for each $S \in \mathcal{S}$, the connected components of $\bG[S]$. The core idea of the algorithm is to gradually grow the connected components of $\bG[S]$ for every $S \in \mathcal{S}$. We proceed in rounds, maintaining a partition of each $\bG[S]$ that refines its connected components. Let $C_S(u)$ denote a unique identifier associated with the part containing node $u$ in $\bG[S]$.

We initialize the process with singleton parts for the nodes in every $\bG[S]$. In each round, for every $S \in \mathcal{S}$ and node $u \in \bG[S]$, we check whether there exists an edge $\{u,w\} \in \bG[S]$ such that $C_S(u) \neq C_S(w)$. Discovering such an edge allows us to merge the respective parts and update $C_S(\cdot)$ accordingly for each $S$. We claim that $\Oh(\log n)$ rounds suffice for these partitions to grow into the connected components. Indeed, any part that does not yet form a complete connected component will, by induction, contain at least $2^i$ nodes after the $i$-th round.

It remains to briefly outline how to find the edges across which we merge parts. We accomplish this in parallel for each $\bG[S]$ via matrix multiplication. To this end, in each round we assign every part in each $\bG[S]$ a value of zero or one, chosen uniformly and independently at random. Let $C'_S(u)$ denote the bit assigned to the part $C_S(u)$.

Note that if, for some $S \in \mathcal{S}$ and $u \in \bG[S]$, there exists an edge $\{u,w\} \in \bG[S]$ with $C_S(u) \neq C_S(w)$, then with probability $1/4$ we have $C'_S(u) = 0$ and $C'_S(w) = 1$. Furthermore, edges contained entirely within a single part of $\bG[S]$ always have their endpoints assigned the same bit. Thus, it suffices to search for an edge spanning parts with differing assigned bits, repeating the procedure a logarithmic number of times to guarantee success with high probability.

The product we perform is between an $|\mathcal{S}| \times n$ Boolean matrix $M$ (indexed by the sets in $\mathcal{S}$ and the vertices in $V$) and the adjacency matrix $A$ of $\bG$ (indexed by the vertices in $V$). For the matrix $M$, we set $M_{S,w} = 1$ if and only if $w \in \bG[S]$ and $C'_S(w) = 1$. It is easy to see that for every $S \in \mathcal{S}$ and $u \in \bG[S]$ with $C'_S(u) = 0$, there exists an inter-part edge $\{u,w\}$ if and only if $(MA)_{S,u} = 1$. The Boolean matrix multiplication witness for this product directly provides the desired edge.

It is also not difficult to derandomize the above procedure. In the full proof, we provide such a derandomized version with a slightly more formal presentation.

\subsection{Recognizing \AT-free Graphs in Matrix Multiplication Time}

We proceed to outline the main idea behind our $\Ohtilde(n^{\omega})$-time algorithm to verify whether a graph $\bG = (V,E)$ contains an \AT. We recall that such a triple consists of an independent set $\{x,y,z\} \subseteq V$ such that there exists a path between every pair of nodes in $\{x,y,z\}$ that does not pass through the neighborhood of the third node.

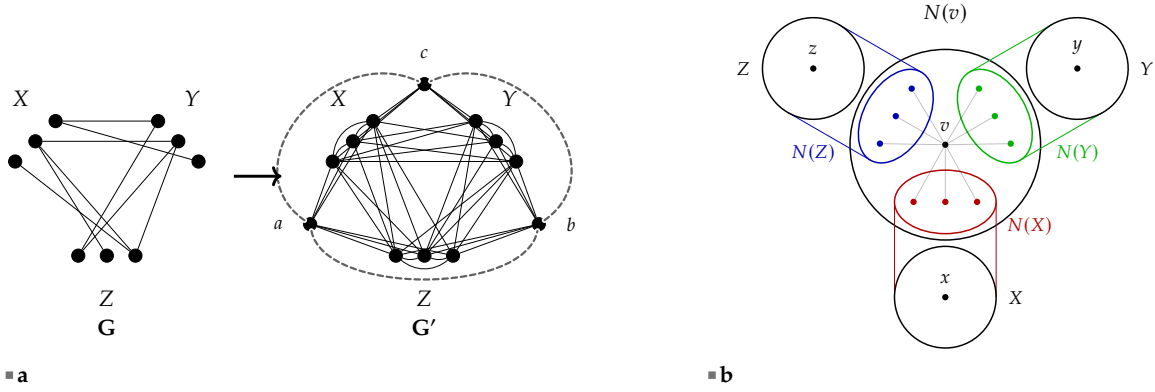
\begin{figure*}[!t]
         \centering
         \begin{subfigure}[b]{0.53\textwidth}
             \centering
             \scalebox{0.7}{\begin{tikzpicture}[scale=0.6, every node/.style={draw,shape=circle, fill, inner sep=2.5pt}]
    % --- FIRST GRAPH (G) ---
    \begin{scope}[scale=0.9]
        % Labels for the tripartitions X, Y, Z
        \node[draw=none, fill=none] at (-3.0, 5.5) {\Large $X$};
        \node[draw=none, fill=none] at (3.0, 5.5) {\Large $Y$};
        \node[draw=none, fill=none] at (0, -1.5) {\Large $Z$};
        
        % Label for the first graph
        \node[draw=none, fill=none] at (0, -2.5) {\Large $\mathbf{G}$};

        \begin{scope}[shift={(-2.5,4)},rotate=45]
            \foreach \i in {1,...,3} {
              \node (L\i) at (1*\i-2.0,0) {};
            }
        \end{scope}
        
        \begin{scope}[shift={(2.5,4)},rotate=-45]
            \foreach \i in {1,...,3} {
              \node (R\i) at (1*\i-2.0,0) {};
            }
        \end{scope}
        
        \foreach \i in {1,...,3} {
          \node (B\i) at (1*\i-2.0,0) {};
        }

        % Original edges between partitions
        \draw (R2) -- (L2);
        \draw[opacity=1] (L3) -- (R1);
        \draw[opacity=1] (R3) -- (L3);

        \draw[opacity=1] (L1) -- (B3);
        \draw[opacity=1] (L2) -- (B2);
        \draw (L2) -- (B3);
        
        \draw[opacity=1] (R1) -- (B1);
        \draw[opacity=1] (R2) -- (B1);
        \draw (R2) -- (B3);
    \end{scope}
    
    % --- SECOND GRAPH (G') ---
    \begin{scope}[shift={(10,0)},scale=0.9]
        % Labels for the tripartitions X, Y, Z
        \node[draw=none, fill=none] at (-3.0, 5.5) {\Large $X$};
        \node[draw=none, fill=none] at (3.0, 5.5) {\Large $Y$};
        \node[draw=none, fill=none] at (0, -1.5) {\Large $Z$};
        
        % Label for the second graph
        \node[draw=none, fill=none] at (0, -2.5) {\Large $\mathbf{G}'$};

        % Spaced out labels a, b, c
        \node[label={[draw=none, fill=none, label distance=2mm]left:$a$}] (a) at (-4.0,1.1) {};
        \node[label={[draw=none, fill=none, label distance=2mm]right:$b$}] (b) at (4.0,1.1) {};
        \node[label={[draw=none, fill=none, label distance=2mm]above:$c$}] (c) at (0,6.0) {};
       
        \begin{scope}[shift={(-2.5,4)},rotate=45]
            \foreach \i in {1,...,3} {
              \node (L\i) at (1*\i-2.0,0) {};
              \draw[opacity=1] (a) -- (L\i);
              \draw[opacity=1] (c) -- (L\i);
            }
            % Curved edges ONLY in the second graph (within X)
            \draw[bend left=20] (L1) to (L2);
            \draw[bend left=20] (L2) to (L3);
            \draw[bend left=40] (L1) to (L3);
        \end{scope}
        
        \begin{scope}[shift={(2.5,4)},rotate=-45]
            \foreach \i in {1,...,3} {
              \node (R\i) at (1*\i-2.0,0) {};
              \draw[opacity=1] (b) -- (R\i);
              \draw[opacity=1] (c) -- (R\i);
            }
            % Curved edges ONLY in the second graph (within Y)
            \draw[bend left=20] (R1) to (R2);
            \draw[bend left=20] (R2) to (R3);
            \draw[bend left=40] (R1) to (R3);
        \end{scope}
        
        \foreach \i in {1,...,3} {
          \node (B\i) at (1*\i-2.0,0) {};
          \draw[opacity=1] (a) -- (B\i);
          \draw[opacity=1] (b) -- (B\i);
        }
        % Curved edges ONLY in the second graph (within Z)
        \draw[bend right=20] (B1) to (B2);
        \draw[bend right=20] (B2) to (B3);
        \draw[bend right=40] (B1) to (B3);

        % Edges between partitions
        \draw[opacity=1] (R1) -- (L1);
        \draw[opacity=1] (R1) -- (L2);
        \draw[opacity=1] (R2) -- (L1);
        \draw[opacity=1] (R2) -- (L3);
        \draw[opacity=1] (R3) -- (L1);
        \draw[opacity=1] (R3) -- (L2);

        \draw[opacity=1] (L1) -- (B1);
        \draw[opacity=1] (L1) -- (B2);
        
        \draw (L2) -- (B1);

        \draw (L3) -- (B1);
        \draw (L3) -- (B2);
        \draw (L3) -- (B3);
        
        \draw[opacity=1] (R1) -- (B2);
        \draw[opacity=1] (R1) -- (B3);
          
        \draw[opacity=1] (R2) -- (B2);

        \draw[opacity=1] (R3) -- (B1);
        \draw[opacity=1] (R3) -- (B2);
        \draw[opacity=1] (R3) -- (B3);

        % Draw the auxiliary triangle last so its edges remain visible.
        \begin{scope}[line cap=round]
            \draw[white, line width=3pt] (a) .. controls (-3.5,-1.4) and (3.5,-1.4) .. (b);
            \draw[white, line width=3pt] (a) to[out=145,in=155,looseness=1.35] (c);
            \draw[white, line width=3pt] (b) to[out=35,in=25,looseness=1.35] (c);
            \draw[gray!75!black, line width=1.2pt, dash pattern=on 3pt off 2pt] (a) .. controls (-3.5,-1.4) and (3.5,-1.4) .. (b);
            \draw[gray!75!black, line width=1.2pt, dash pattern=on 3pt off 2pt] (a) to[out=145,in=155,looseness=1.35] (c);
            \draw[gray!75!black, line width=1.2pt, dash pattern=on 3pt off 2pt] (b) to[out=35,in=25,looseness=1.35] (c);
        \end{scope}
    \end{scope}

    % Draw the transformation arrow in the foreground, with a white halo
    % so graph edges remain visually separate from it.
    \draw[white, line width=4pt] (4.0, 2.5) -- (5.5, 2.5);
    \draw[->, ultra thick] (4.0, 2.5) -- (5.5, 2.5);
\end{tikzpicture}}
             \caption{}
             \label{fig:attriple:a}
         \end{subfigure}
         \hfill
        \begin{subfigure}[b]{0.41\textwidth}
             \centering
             \scalebox{0.7}{\begin{tikzpicture}[scale=1.2]

% ---- Central circle N(v) ----
\draw[thick] (0,0) circle (1.5cm);
\node (v) at (0,0) [circle,fill,inner sep=1.2pt] {};
\node[above=2pt of v] {$v$};
\node[above=60pt of v] {\large $N(v)$};

% ---- BLUE ----

\begin{scope}[rotate around={-120:(0.0,0.0)}]
\draw[thick, blue!70!black] (0,-0.9) ellipse (0.8cm and 0.5cm);
\foreach \p [count=\i] in {(-0.5,-0.9), (0,-0.9), (0.5,-0.9)}{
   \node (b\i) at \p [circle, fill, inner sep=1.2pt, blue!70!black] {};
    \draw[opacity=0.25] (v) -- (b\i);
}

\draw[blue!70!black] (-0.8,-2.4) -- (-0.8,-0.9);
\draw[blue!70!black] (0.8,-2.4) -- (0.8,-0.9);
\draw[thick] (0,-2.4) circle (0.8cm);
\node (z) at (0,-2.4) [circle,fill,inner sep=1.2pt] {};
\end{scope}

% ---- GREEN ----

\begin{scope}[rotate around={120:(0.0,0.0)}]
\draw[thick, green!70!black] (0,-0.9) ellipse (0.8cm and 0.5cm);
\foreach \p [count=\i] in {(-0.5,-0.9), (0,-0.9), (0.5,-0.9)}{
   \node (g\i) at \p [circle, fill, inner sep=1.2pt, green!70!black] {};
   \draw[opacity=0.25] (v) -- (g\i);
}

\draw[green!70!black] (-0.8,-2.4) -- (-0.8,-0.9);
\draw[green!70!black] (0.8,-2.4) -- (0.8,-0.9);

\draw[thick] (0,-2.4) circle (0.8cm);
\node (y) at (0,-2.4) [circle,fill,inner sep=1.2pt] {};
\end{scope}

% ---- RED ----
\begin{scope}
\draw[thick, red!70!black] (0,-0.9) ellipse (0.8cm and 0.5cm);
\foreach \p [count=\i] in {(-0.5,-0.9), (0,-0.9), (0.5,-0.9)}{
   \node (r\i) at \p [circle, fill, inner sep=1.2pt, red!70!black] {};
   \draw[opacity=0.25] (v) -- (r\i);
}

\draw[red!70!black] (-0.8,-2.4) -- (-0.8,-0.9);
\draw[red!70!black] (0.8,-2.4) -- (0.8,-0.9);

\draw[thick] (0,-2.4) circle (0.8cm);
\node (x) at (0,-2.4) [circle,fill,inner sep=1.2pt] {};
\end{scope}

\node[above=2pt of z] {$z$};
\node[left=of z] {$Z$};
\node[blue!70!black, below=35ptof z] {$N(Z)$};

\node[above=2pt of y] {$y$};
\node[right=of y] {$Y$};
\node[green!70!black, below=35pt of y] {$N(Y)$};

\node[above=2pt of x] {$x$};
\node[right=of x] {$X$};
\node[red!70!black, above right=of x] {$N(X)$};

\end{tikzpicture}}
             \caption{}
             \label{fig:attriple:b}
        \end{subfigure}

        \caption{
            Visualization of the lower bound construction and a fact that we use in the $\Ohtilde(n^{\omega})$-time algorithm for asteroidal triple detection. In (a), dashed gray edges join the auxiliary vertices $a,b,c$.
         }
         \label{fig:attriple}
\end{figure*}

Before presenting the main idea of the algorithm, we note that the triangle-detection hardness of finding such a triple can be shown to rely solely on the condition that $\{x,y,z\} \subseteq V$ forms an independent set. To clarify our point, we consider the following elementary reduction demonstrating how, given a tripartite graph $\bG = (X \cup Y \cup Z, E)$, we can build a graph $\bG' = (V', E')$ such that $\bG$ contains a triangle $\{x,y,z\}$ if and only if $\bG'$ contains an \AT $\{x',y',z'\}$\footnote{It is well known that triangle detection on an $n$-node graph can be reduced to triangle detection on a tripartite $3n$-node graph: simply create three copies of the vertex set and add edges only between copies of adjacent nodes that reside in different vertex partitions.} (see \cref{fig:attriple:a}). To achieve this, it suffices to invert all edges in $\bG$ (including those within $X$, $Y$, and $Z$, which consequently become cliques) and add three new vertices $a,b,c$ that are fully connected to $X \cup Z$, $Y \cup Z$, and $X \cup Y$, respectively. We also add the edges $\{a,b\}$, $\{b,c\}$, and $\{c,a\}$. It is easy to see that $\bG'$ satisfies the claimed property. However, notice that for every triple with $x \in X$, $y \in Y$, and $z \in Z$, there is always an $x$-$y$ path avoiding $N[z]$ (and similarly for the other pairs). This means that the reduction relies solely on the fact that an \AT $\{x,y,z\}$ is an independent set, as the path condition holds vacuously. The open question was thus whether this path condition which is not exploited in the lower bound could introduce additional computational hardness.

Moving on to the algorithm, we first illustrate Kratsch and Spinrad's approach. Given a graph $\bG = (V,E)$, define for each $v \in V$ the graph $\bG_v \coloneqq \bG[V \setminus N[v]]$, obtained by removing the closed neighborhood of $v$. As a first step, Kratsch and Spinrad compute the neighborhood-deleted components, i.e., the connected components of each $\bG_v$, and then construct a graph $\tilde{\bG}$ as follows:\footnote{Technically, the graph $\tilde{\bG}$ we construct is the knotting graph $K(\bG)$ of $\bG$. Knotting graphs were introduced by Gallai \cite{G67}, who showed that $\bG$ is a comparability graph if and only if $K(\bG)$ is bipartite.}
\begin{enumerate}
    \item for each $v \in V$ and each connected component $C$ of $\bG_v$, they introduce a node in $\tilde{\bG}$; and
    \item for each non-edge $\{u,v\}$ of $\bG$, they add an edge $\{C_u(v), C_v(u)\}$ in $\tilde{\bG}$, where $C_x(y)$ now denotes the connected component containing $y$ in $\bG_x$.
\end{enumerate}
It is not difficult to see that $\bG$ contains an \AT if and only if $\tilde{\bG}$ contains a triangle. This allows us to solve the problem by running a triangle detection algorithm on $\tilde{\bG}$. However, due to the number of nodes in $\tilde{\bG}$, we cannot use standard triangle detection; instead, we must resort to sparse triangle detection, which runs in time $\Oh(m^{2\omega/(\omega+1)})$, where $m$ is the number of edges \cite{AYZ97}. Since $\tilde{\bG}$ has $\Oh(n^2)$ edges, this running time becomes $\Oh(n^{4\omega/(\omega+1)}) = \Oh(n^{2.81})$.\footnote{To be precise, the observation regarding the construction of this graph is due to Lekkerkerker and Boland \cite{LB62}, while the use of sparse matrix multiplication was introduced by K\"ohler \cite{K04}. However, Kratsch and Spinrad were the first to achieve an overall subcubic running time, by computing the neighborhood-deleted components subcubically.}

The first key fact we prove is that if $\tilde{\bG}$ contains a triangle $\{C_1, C_2, C_3\}$ and $C_1'$ is a connected component containing exactly the same nodes as $C_1$, then $\{C_1', C_2, C_3\}$ is also a triangle (even if $C_1$ and $C_1'$ come from different graphs $\bG_v$). Thus, for each distinct connected component across all graphs $\bG_x$, it suffices to retain only one representative in $\tilde{\bG}$.

In a general graph, the number of distinct components in $\tilde{\bG}$ can be as large as $\Omega(n^2)$; however, we only need to consider \AT-free graphs. That is, if we can prove an upper bound of $\alpha(n)$ on the number of distinct components for \AT-free graphs, we can first check whether the actual number of distinct components exceeds this bound. If it does, we immediately report that an \AT triple exists; otherwise, we run a triangle detection instance on $\Oh(\alpha(n))$ nodes. It turns out---and we believe this is a highly interesting structural result of independent interest---that such a bound is indeed possible for \AT-free graphs.

\begin{restatable*}{theorem}{atfreesmall} \label{thm:atfree_small}
    Let $\bG = (V, E)$ be an \AT-free graph.
    Moreover, let $\bG_x \coloneqq \bG[V \setminus N[x]]$,
    and let $\mC_x = \{C_{x}(v)\}_{v \in \bG_x}$ be the set of all connected components of $\bG_x$. Then there are at most $\Oh(n \log n)$ distinct components across all $\bG_x$; that is, $|\mC| \leq 400\, n \log n$ for $\mC \coloneqq \bigcup_{x \in V} \mC_{x}$.
\end{restatable*}

Although the full proof is one of the more involved parts of our paper, it rests on a simple observation. The following basic fact serves as the first building block of the proof and provides intuition for why the absence of an \AT yields the desired bound.

Consider a node $v \in V$ and three distinct connected components $X$, $Y$, and $Z$ of $\bG_v$. Furthermore, consider three nodes $x \in X$, $y \in Y$, and $z \in Z$, together with the out-neighborhoods $N(X)$, $N(Y)$, and $N(Z)$, which consist entirely of nodes in $N(v)$ (see \cref{fig:attriple:b}). Observe that if $\{x, y, z\}$ forms an \AT, then all of the following conditions must hold: $N(X) \setminus N(z) \neq \emptyset$, $N(X) \setminus N(y) \neq \emptyset$, $N(Z) \setminus N(x) \neq \emptyset$, $N(Z) \setminus N(y) \neq \emptyset$, $N(Y) \setminus N(x) \neq \emptyset$, and $N(Y) \setminus N(z) \neq \emptyset$; otherwise, the neighborhood-disjoint paths required for the \AT could not exist. Consequently, if $\bG$ is \AT-free, at least one of these conditions must fail. Without loss of generality, assume that $N(X) \setminus N(z) = \emptyset$, which is equivalent to $N(X) \subseteq N(z)$. It is not difficult to see that this implies $X$ is also a connected component of $\bG_z$. In this sense, being \AT-free restricts the number of distinct components, though establishing the full bound requires a more involved argument.

\subsection{Listing Dominating Pairs in Matrix Multiplication Time}

Next, we turn to listing all dominating pairs $u,v$ of an input graph $\bG$. Recall that $u,v$ form a dominating pair if every $u$-$v$ path is a dominating set, and recall that $C_x(y)$ denotes the connected component containing $y$ in $\bG_x$. The connection between this problem and neighborhood-deleted components is that $u,v$ is \emph{not} a dominating pair if and only if there is some node $w \notin N[u]\cup N[v]$ such that $C_w(u) = C_w(v)$; that is, there exists a $u$-$v$ path avoiding the neighborhood of $w$. In such a case, we say that the component $C_w(u) = C_w(v)$ \emph{witnesses} the non-dominating pair $u,v$.

Following this logic, a naive algorithm would simply iterate over each $v \in V$, examine each connected component $C$ of $\bG_v$, and mark all pairs of distinct nodes $x,y \in C$ as non-dominating. We note that for a subset $\mathcal{C}$ of connected components (potentially spanning different $\bG_v$), we can implement this in time $\Oh(\mathsf{MM}(n, |\mathcal{C}|, n))$. To achieve this, it suffices to construct an $n \times |\mathcal{C}|$ matrix $M$, indexed by vertices in $V$ and components in $\mathcal{C}$ respectively, with $M_{v,C} = 1$ if $v \in C$. The Boolean product $M M^{\top}$ indicates exactly when a pair $u,v$ is witnessed as non-dominating by some component in $\mathcal{C}$, which occurs precisely when $(M M^{\top})_{u,v} = 1$. Kratsch and Spinrad's approach distinguishes between large components, which contain at least $\Delta$ nodes, and small ones. Since there can be at most $n^2 / \Delta$ of the former across all $\bG_v$, marking all their pairs as non-dominating via the matrix method takes $\mathsf{MM}(n, n^2 / \Delta, n) = \Oh(n^{\omega+1} / \Delta)$ time. They also show how to process small components in $\Oh(n^2 \cdot \Delta)$ time; balancing these two bounds yields an overall running time of $\Oh(n^{(\omega+3)/2})$.

Our algorithm differs in two ways:
\begin{enumerate}
    \item Instead of collecting into the set $\mathcal{C}$ all connected components above a threshold size $\Delta$, we include only the largest connected component from each $\bG_v$. By doing so, it takes only $\Oh(n^{\omega})$ time to find all non-dominating pairs witnessed by these components.
    \item On the other hand, for non-dominating pairs that have no witness in $\mathcal{C}$, we prove an interesting structural result: if $u,v$ has no ``large'' witness, then it must have many witnesses. More specifically, we prove that if $u,v$ is a non-dominating pair witnessed by some $C \notin \mathcal{C}$ in $\bG_x$, then there are at least $\Omega(|V \setminus N[x]|)$ nodes $y \in V$ such that $\bG_y$ contains a component witnessing $u,v$. Assuming for simplicity that $|N[x]| = o(n)$ for all $x \in V$, we can exploit this property as follows: we sample $\Theta(\log n)$ nodes and, for each sampled node $y$, find all non-dominating pairs witnessed by any component of $\bG_y$. This takes $\Oh(n^2)$ time per sampled node $y$, since the sizes of the components in each $\bG_y$ sum to $\Oh(n)$.
\end{enumerate}
In the full proof, we again provide a derandomized version of the algorithm.

\subsection*{Open Problems}

To conclude this section, we list some other open problems:

\begin{itemize}
    \item First, we report an open problem posed in \cite{KS06}. It is well established that every \AT-free graph contains a dominating pair \cite{COS97}, which can, in fact, be found in linear time \cite{COS99}. However, neither \AT-free graph recognition nor finding a dominating pair in a general graph can be solved in linear time unless $\omega = 2$ or some fine-grained complexity hypothesis fails. This leads to an interesting open question: is there a robust, linear-time algorithm for finding a dominating pair in an \AT-free graph?  Here, ``robust'' is used in the sense of \cite{RS03}; a robust linear-time algorithm must take an arbitrary (i.e., not necessarily \AT-free) graph as input and, in linear time, either output a dominating pair or correctly report that the graph is not \AT-free.

    \item We showed that detecting an \AT-triple has the same time complexity as detecting a triangle in a graph. In time $\Oh(n^{\omega})$, however, we can solve problems stronger than triangle detection, e.g., the all-edge triangle problem, where for each edge we want to know whether it is contained in a triangle. Our approach for detecting an \AT-triple does not seem to yield such an all-edge version in $\Oh(n^{\omega})$ time, and the latter seems considerably more difficult: namely, for each pair $x,y$ we want to know whether there is a $z$ such that $\{x,y,z\}$ is an \AT. It would be interesting to determine whether a running time of $\Oh(n^{\omega})$ is achievable, or whether some hardness result rules this out.

    \item Generally, there is a plethora of other important structures in structural graph theory that admit polynomial-time algorithms which are nonetheless clearly far from optimal, where a fine-grained complexity lens could help uncover new 
    structural properties; see, e.g., perfect graphs~\cite{CLV03, CKS05} or graphs 
    with even or odd holes~\cite{CH15, CSSP20, CLL22}.
\end{itemize}

\section{Preliminaries}
\label{sec:prelims}

\subparagraph*{Set notation.}
For integers \(i, j \in \mathbb{Z}\), we write \(\fragment{i}{j}\) to represent the set \(\{i, \dots, j\}\), and \(\fragmentco{i}{j}\) to denote the set \(\{i, \dots, j - 1\}\).
We define \(\fragmentoc{i}{j}\) and \(\fragmentoo{i}{j}\) similarly.

\subparagraph*{Graph notation.} All graphs considered in this paper are undirected. For a graph $\bG = (V, E)$, we denote by $n := |V|$ the number of nodes and by $m := |E|$ the number of edges. For a node $v \in V$ we indicate with $N(v)$ its neighborhood, and with $N[v] \coloneqq N(v) \cup \{v\}$ its \emph{closed neighborhood}. Similarly, we define the neighborhood of a set $S \subseteq V$ as $N(S) := \cup_{v \in S} N(v) \setminus S$, and its closed neighborhood as $N[S] := N(S) \cup S$. In case of ambiguity due to the existence of multiple graphs, we write as subscript the graph which the neighborhood refers to, e.g., $N_{\bG}(v)$.

We indicate with $\bG_v$ the graph obtained by removing from $\bG$ the closed neighborhood, i.e., $\bG_v \coloneqq \bG[V \setminus N[v]]$. Lastly, for a node $u \in V \setminus N[v]$, we let $C_v(u)$ be the connected component of $\bG_v$ that contains $u$. More specifically, $C_v(u)$ is the set of nodes in $\bG_v$ that are in the same connected component as $u$.

Finally, we say that a (simple) $x$-$y$ path in $\bG = (V,E)$ avoids a set $S \subseteq V$, if it contains no node in $S$.

\subparagraph*{Matrix Multiplication and Witnesses.} For our algorithms we need to be able to list \emph{witnesses} for a Boolean matrix product. For that purpose, we use the following result.

\begin{theorem}[Deterministic Witness Reporting \cite{AN96}]
    \label{thm:det_witnesses}
    There is a deterministic $\Ohtilde(\mathsf{MM}(p, q, r))$-time algorithm that, given a Boolean $p \times q$ matrix $A$ and a $q \times r$ matrix $B$, reports one witness for each nonzero entry of the Boolean matrix product $AB$. That is, for every pair $(i, j)$ such that $AB[i,j] = 1$, the algorithm reports an index $k \in \fragment{1}{q}$ satisfying $A[i,k] = 1$ and $B[k,j] = 1$. \lipicsEnd
\end{theorem}

\subparagraph*{Deterministic Fingerprints for Partitions.} 

Suppose we are given a graph $\bG = (V, E)$ and a family of vertex sets $\mathcal{S} \subseteq 2^V$. Further, let $C_S(v)$ be the connected component of $v \in V$ in $\bG[S]$. To quickly check whether $C_{S}(v) = C_{S'}(v')$ for distinct $v,v' \in V$ and $S, S' \in \mathcal{S}$, we introduce the following fingerprints for the connected components.

\begin{lemma}\label{lem:fingerprints}
    Consider a set $V = \fragment{1}{n}$ together with $k$ partitions
    $\mathcal{P}_1, \ldots, \mathcal{P}_k$ of $V$, and for $v \in V$ let $P_i(v)$ denote
    the set of the $i$-th partition containing $v$. Then, in deterministic time
    $\Oh(nk \log(nk))$, we can construct functions
    $h_1, \ldots, h_k : V \rightarrow \fragment{1}{M}$ with $M = \Oh(nk)$ that
    assigns a \emph{fingerprint} to each set of each partition, such that
    $h_i(v) = h_j(u)$ if and only if $P_i(v) = P_j(u)$.
\end{lemma}
\begin{proof}
    Let $\mathcal{P} = \biguplus_{i=1}^{k} \mathcal{P}_i$ be the multiset containing
    all sets. We proceed as follows.
    
    \subparagraph*{Our algorithm.}
    First, we sort the elements of each set
    $P \in \mathcal{P}$, looking at them as lists.
    Equal lists have equal size, so we partition $\mathcal{P}$ into buckets by length; lists in different buckets get different fingerprints. Within a bucket $\mathcal{P}^\ell \subseteq \mP$ all lists have length $\ell$, and we sort them lexicographically, making equal lists adjacent. A single
    scan then assigns fingerprints by incrementing a global counter whenever the
    current list differs from its predecessor, and we store $h_i(v)$ as the
    fingerprint of $P_i(v)$ in a table indexed by $\fragment{1}{k} \times V$. 
    It is not difficult to see that this construction satisfies the claimed property.

    \subparagraph*{Runtime.}
    Sorting sets initially, bucketing by lengths, and the final scan cost at most $\Oh(nk \log n)$. 

    For the lexicographic sort on $\mP^\ell$ we use $ \Oh( |\mP^\ell| \log |\mP^\ell|)$ comparisons, where each comparison costs $\Oh(\ell)$ time. Overall this costs $\Oh(\sum_{\ell=1}^{n} \ell\ \cdot |\mathcal{P}^\ell| \log  |\mathcal{P}^\ell|) = \Oh(nk \log(nk))$, because the sizes of the lists sum up to $nk$ and $|\mP^{\ell}| = \Oh(nk)$; empty buckets contribute zero.
\end{proof}

\section{Proofs}

\subsection{Induced Subgraph Connected Components in Matrix Multiplication Time}

In this subsection, we present the proof for \cref{lem:ndcc_to_bmm} and \cref{lem:ndcc_lb}, which proves the upper and lower bounds for the induced subgraph connected component problem.

\ndcctobmm

\begin{proof}
    Let $\bG = (V,E)$ be an undirected $n$-node graph, and let $\mathcal{S} \subseteq 2^V$ be non-empty.
    To compute the connected components of $\bG[S]$ for each $S \in \mathcal{S}$, we rely on the following subroutine.

    \begin{claim}\label{clm:ndcc_to_bmm}
        For each $S \in \mathcal{S}$, let $f_S : V \rightarrow \fragment{1}{M} \cup \{\bot\}$ be a function, where $M = \poly(n)$. 
        Then, in time $\Ohtilde(\mathsf{MM}(|\mathcal{S}|, n, n))$, for each $S \in \mathcal{S}$ and $w$ such that $f_{S}(w) \neq \bot$, we can find a node $u \in V$ that satisfies $\{u,w\} \in E$, $f_S(u) \neq f_{S}(w)$ and $f_{S}(u) \neq \bot$, or verify that no such node exists.
    \end{claim}
    \begin{claimproof}
        Let $B = \Oh(\log n)$ be the number of bits required to represent any value in $\fragment{1}{M}$.
        The algorithm proceeds as follows:
        \begin{enumerate}[(i)]
            \item For each $i \in \fragment{1}{B}$, we construct two Boolean $|\mathcal{S}| \times n$ matrices $M^{b}_i$ for $b \in \{0, 1\}$, indexed by sets in $\mathcal{S}$ (rows) and by nodes in $V$ (columns), defined as follows:
            \[
                M_i^b[S,u] \coloneqq
                \begin{cases}
                    1 & \text{if $f_S(u) \neq \bot$ and the $i$-th bit of $f_S(u)$ is $b$,} \\
                    0 & \text{otherwise.}
                \end{cases}
            \]
            \item For each $i \in \fragment{1}{B}$, we use \cref{thm:det_witnesses}
            to compute the two BMM products $M^0_i A$ and $M^1_i A$ in time $\Oh(\mathsf{MM}(|\mathcal{S}|, n, n))$, where $A$ denotes the adjacency matrix of $\bG$. We also compute witnesses for all non-zero entries in these products.
            \item To find a suitable $u$ for each $S$ and $w$ such that $f_S(w) \neq \bot$, we iterate as follows: for each $i \in \fragment{1}{B}$, let $b_i$ be the $i$-th bit in the binary representation of $f_S(w)$. If for any $i \in \fragment{1}{B}$ we have $(M^{1 - b_i}_i A)[S,w] = 1$, we report the corresponding witness node of the product $(M^{1 - b_i}_i A)[S,w]$ as the node $u$ for $S,w$.
        \end{enumerate}

        Regarding correctness, fix $S$ and $w$ such that $f_S(w) \neq \bot$, and let $b_i$ be defined as above. Note that for each $i \in \fragment{1}{B}$, we have $(M^{1 - b_i}_i A)[S,w] = 1$ if and only if there exists a node $u$ such that $\{u,w\} \in E$, $f_S(u) \neq \bot$, and the $i$-th bit of $f_S(u)$ differs from the $i$-th bit of $f_S(w)$.  From this, the correctness is immediate: we only report a node $u$ if $f_S(u) \neq f_S(w)$. Conversely, if such a $u$ exists, then $f_S(u)$ and $f_S(w)$ must differ in at least one bit, ensuring that we find such $u$.

        Regarding the runtime, we use $\Oh(\log n)$ many times \cref{thm:det_witnesses}; all other operations can be implemented altogether in time $\Oh(|\mathcal{S}|n \log n)$.
    \end{claimproof}

    \subparagraph{Our algorithm.}
    With the subroutine setup in place, we now give the full algorithm.
    \begin{enumerate}[(i)]
        \item For each $S \in \mathcal{S}$, initialize a union-find (disjoint-set) data structure\footnote{A union-find (disjoint-set) data structure $D$ maintains a partition of a universe into disjoint sets, supporting \textsc{Find}$_D(x)$ (returns the representative of $x$’s set) and \textsc{Union}$_D(x,y)$ (merges the two sets), with near-constant amortized time $\Oh(\alpha(n))$ using union-by-rank and path compression \cite{T75}. Here $\alpha(n)$ denotes the inverse-Ackermann function.} $D_S$ over $V$ (equivalently, over the universe $\fragment{1}{n}$).
        The invariant we maintain is that every set in $D_S$ is contained in a connected component of $\bG[S]$ (or is a singleton corresponding to a node not in $\bG[S]$).
    
        \item We run for $\Theta(\log n)$ rounds. In the $i$-th round, for each $S \in \mathcal{S}$ and $w \in V$, we define
        \[
          f_S(w) \;=\;
          \begin{cases}
            \textsc{Find}_{D_S}(w), & w \in \bG[S],\\
            \bot, & w \notin \bG[S].
          \end{cases}
        \]
        We then invoke the subroutine from \cref{clm:ndcc_to_bmm} to find, for each $S$ and $w$, a witness $u$ as specified there (or determine that no such $u$ exists).
        Whenever we obtain such a witness, we perform $\textsc{Union}_{D_S}(u,w)$. 
        The invariant is preserved because \cref{clm:ndcc_to_bmm} guarantees that $(u,w)\in E$ and that both $u,w \in \bG[S]$. After processing all such unions, we proceed to the next round.
    
        \item For each $S \in \mathcal{S}$, output all sets in $D_S$ except singleton sets corresponding to nodes $w \notin \bG[S]$; these are the connected components of $\bG[S]$.
    \end{enumerate}
    
    \subparagraph*{Correctness.}
    It remains to show that after $\Theta(\log n)$ rounds we have that the sets we maintain are equal to the connected components.
    For each $S \in \mathcal{S}$, let $m_S$ denote the number of sets in $D_S$ that are strict subsets of connected components of $\bG[S]$.
    In each round, $m_S$ decreases by at least a factor of $2$: every set counted by $m_S$ has at least one outgoing edge to another set in the same connected component, and the corresponding merge ensures that in the next round it is combined with another set.
    Hence after $\Theta(\log n)$ rounds, $m_S=0$ for every $S$, so each set in $D_S$ exactly matches a connected component of $\bG[S]$.

    \subparagraph*{Runtime.} In each of the $\Theta(\log n)$ rounds we use \cref{clm:ndcc_to_bmm} and perform at most $\Oh(|\mathcal{S}| n)$ union operations. This gives an overall runtime of $\Ohtilde(\mathsf{MM}(|\mathcal{S}|, n, n))$.
\end{proof}

\ndcclb

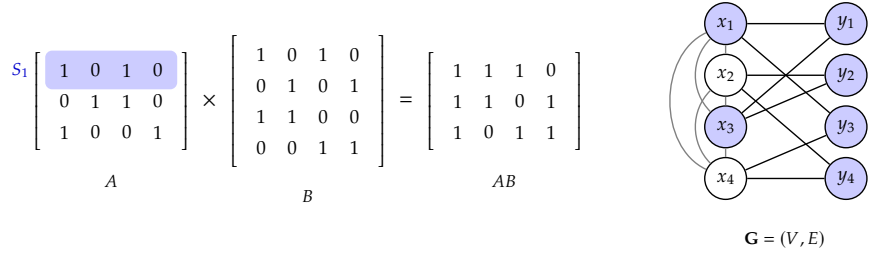
\begin{figure}[t]
        \centering
        \scalebox{0.65}{\begin{tikzpicture}[
    node distance = 1.5cm,
    xnode/.style={circle, draw=black, thick, minimum size=8.5mm, inner sep=0pt, font=\large},
    ynode/.style={circle, draw=black, thick, minimum size=8.5mm, inner sep=0pt, font=\large},
    clique_edge/.style={gray, thick},
    b_edge/.style={black, thick}
]

% Declare background layer for the row highlighting
\pgfdeclarelayer{background}
\pgfsetlayers{background,main}

% ====================
% Wrapper scope for left shift
% ====================
\begin{scope}[xshift=-1.5cm]

\matrix (matAB) [
    matrix of math nodes,
    left delimiter={[}, right delimiter={]},
    nodes={minimum width=1.8em, minimum height=1.8em, anchor=center, font=\large}
] at (8, 0) {
    1 & 1 & 1 & 0 \\
    1 & 1 & 0 & 1 \\
    1 & 0 & 1 & 1 \\
};
\node[below=0.3cm of matAB] {$AB$};

% ====================
% MATRIX A (3x4)
% ====================
\matrix (matA) [
    matrix of math nodes,
    left delimiter={[}, right delimiter={]},
    nodes={minimum width=1.8em, minimum height=1.8em, anchor=center, font=\large}
] at (0, 0) {
    1 & 0 & 1 & 0 \\
    0 & 1 & 1 & 0 \\
    1 & 0 & 0 & 1 \\
};
\node[below=0.3cm of matA] {$A$};

% Highlight Row 1 of A to represent S_1
\begin{pgfonlayer}{background}
    \node[fit=(matA-1-1) (matA-1-4), fill=blue!20, rounded corners, inner sep=2pt] (row1) {};
\end{pgfonlayer}
\node[left=0.2cm of row1, text=blue!80!black, font=\bfseries] {$S_1$};

\matrix (matB) [
    matrix of math nodes,
    left delimiter={[}, right delimiter={]},
    nodes={minimum width=1.8em, minimum height=1.8em, anchor=center, font=\large}
] at (4, 0) {
    1 & 0 & 1 & 0 \\
    0 & 1 & 0 & 1 \\
    1 & 1 & 0 & 0 \\
    0 & 0 & 1 & 1 \\
};
\node[below=0.3cm of matB] {$B$};

% ====================
% Multiplication Sign
% ====================
\node[at=($(matA.east)!0.5!(matB.west)$)] (times) {\Large$\times$};

% ====================
% MATRIX B (4x4)
% ====================

% ====================
% Equals Sign
% ====================
\node[at=($(matB.east)!0.5!(matAB.west)$)] (equals) {\Large$=$};

% ====================
% MATRIX AB (3x4) - The Boolean Product
% ====================

% ====================
% GRAPH G = (V,E)
% ====================
\begin{scope}[shift={(12.5, 0)}, scale=0.7]
    % X nodes (Since A[1,1]=1 and A[1,3]=1, we highlight x1 and x3)
    \node[xnode, fill=blue!20] (x1) at (0,  2.25) {$x_1$};
    \node[xnode]               (x2) at (0,  0.75) {$x_2$};
    \node[xnode, fill=blue!20] (x3) at (0, -0.75) {$x_3$};
    \node[xnode]               (x4) at (0, -2.25) {$x_4$};

    % Y nodes
    \node[ynode, fill=blue!20] (y1) at (3.5,  2.25) {$y_1$};
    \node[ynode, fill=blue!20] (y2) at (3.5,  0.75) {$y_2$};
    \node[ynode, fill=blue!20] (y3) at (3.5, -0.75) {$y_3$};
    \node[ynode, fill=blue!20] (y4) at (3.5, -2.25) {$y_4$};

    % Clique edges for {x1, x2, x3, x4}
    % Drawn to the left (bend right relative to the downward drawing direction) to avoid crossing B edges
    \draw[clique_edge] (x1) -- (x2);
    \draw[clique_edge] (x2) -- (x3);
    \draw[clique_edge] (x3) -- (x4);
    \draw[clique_edge] (x1) to[bend right=45] (x3);
    \draw[clique_edge] (x2) to[bend right=45] (x4);
    \draw[clique_edge] (x1) to[bend right=65] (x4);

    % Bipartite edges between X and Y based strictly on matrix B[k, j] = 1
    % Row 1 of B: indices 1, 3
    \draw[b_edge] (x1) -- (y1);
    \draw[b_edge] (x1) -- (y3);
    
    % Row 2 of B: indices 2, 4
    \draw[b_edge] (x2) -- (y2);
    \draw[b_edge] (x2) -- (y4);
    
    % Row 3 of B: indices 1, 2
    \draw[b_edge] (x3) -- (y1);
    \draw[b_edge] (x3) -- (y2);
    
    % Row 4 of B: indices 3, 4
    \draw[b_edge] (x4) -- (y3);
    \draw[b_edge] (x4) -- (y4);
\end{scope}

\node[below=0.5cm of x4, xshift=1.2cm] {\textbf{$\mathbf{G} = (V, E)$}};

\end{scope} % End of shifting scope
\end{tikzpicture}}
        \caption{An example for the construction in \cref{lem:ndcc_lb}}
        \label{fig:ndcc_lb}
\end{figure}

\begin{proof}
    Based on $B$, we construct the graph $\bG = (V, E)$ as follows (for an example, refer to \cref{fig:ndcc_lb}):
    \begin{itemize}
        \item $V$ contains $2n$ vertices, denoted by $x_1, \ldots, x_n, y_1, \ldots, y_n$; and
        \item $E$ contains edges such that the nodes $\{x_1, \ldots, x_n\}$ are fully connected (i.e., they form a clique). Moreover, an edge $\{x_k, y_j\}$ exists if and only if $B[k, j] = 1$.
    \end{itemize}
    Based on $A$, we construct the family of sets $\mathcal{S}$ as follows: for each row $i \in \fragment{1}{s}$, we add a set $S_i$ to $\mathcal{S}$ containing all nodes $x_k$ such that $A[i,k] = 1$. We also add to $S_i$ all nodes in $\{y_1, \ldots, y_n\}$. That is, $S_i = \{x_k \mid k \in \fragment{1}{n}, A[i,k] = 1\} \cup \{y_1, \ldots, y_n\}$.

    Note that by the construction of the clique, for each $S_i \in \mathcal{S}$, there is a unique connected component in $\bG[S_i]$ that contains all nodes in $\{x_1, \ldots, x_n\} \cap S_i$. We call this the \emph{giant component} of $\bG[S_i]$. (Without loss of generality, we assume all $\{x_1, \ldots, x_n\} \cap S_i$ are non-empty, otherwise the corresponding row is the zero vector and computing the product for that row is trivial.) 
    
    To infer the product $AB$, we make the following claim: $AB[i,j] = 1$ if and only if the giant component of $\bG[S_i]$ is non-empty and contains $y_j$. The only if direction is straightforward: if there is a $k \in \fragment{1}{n}$ such that $A[i,k] = 1$ and $B[k, j] = 1$, then $x_k \in S_i$ and the edge $\{x_k, y_j\} \in E$. Therefore, $y_j$ is in the giant component. For the if direction, assume that $y_j$ is in the giant component. This is only possible if there is a $k \in \fragment{1}{n}$ such that $\{x_k, y_j\} \in E$, which is equivalent to $B[k, j] = 1$. Since $x_k \in S_i$, we must also have $A[i,k] = 1$. Thus, we have $AB[i,j] = 1$.
\end{proof}

\subsection{Recognizing \AT-free Graphs in Matrix Multiplication Time}

In this section, we present an $\Ohtilde(n^{\omega})$-time algorithm for recognizing \AT-free graphs. To achieve this, we first introduce a helpful relationship between neighborhood-deleted components that share the same vertices across different $\bG_x$.

\begin{proposition}\label{prp:block}
    Let $G = (V, E)$ be an undirected graph. Consider a vertex $x \in V$ and two vertices $a, b \in \bG_x$ such that $a$ and $b$ belong to distinct neighborhood-deleted components of $x$, i.e., $C_{x}(a) \neq C_{x}(b)$. The following conditions are equivalent:
    \begin{enumerate}[(1)]
        \item The neighborhood-deleted component containing $a$ is identical in $\bG_x$ and $\bG_b$, i.e., $C_{x}(a) = C_{b}(a)$.
        \label{it:block:1}
        \item The neighborhood $N(b)$ contains the neighborhood of $C_{x}(a)$. 
        \label{it:block:2}
    \end{enumerate}
\end{proposition}
\begin{proof}
    We first observe that $C_{x}(a) \subseteq C_{b}(a)$ always holds. Since $C_x(a) \neq C_x(b)$, no vertex in $C_x(a)$ is adjacent to $b$; thus, $N[b] \cap C_{x}(a) = \emptyset$ and all paths within $C_x(a)$ remain intact in $\bG_b$.

    Suppose Condition~\ref{it:block:2} fails, i.e., there exists $v \in N(x) \cap N(C_{x}(a))$ such that $v \notin N(b)$. Since $v$ is adjacent to $C_x(a)$ and is not removed in $\bG_b$, we have $v \in C_b(a)$. However, $v \notin C_x(a)$ (as $v \in N(x)$), which implies $C_x(a) \neq C_b(a)$, so Condition~\ref{it:block:1} fails.
    
    Conversely, assume Condition~\ref{it:block:2} holds. Every vertex $v \in N(x) \cap N(C_{x}(a))$ is removed in $\bG_b$ because $v \in N(b)$. Since any path from $C_x(a)$ to the rest of the graph must pass through these boundary vertices in $N(x)$, their removal isolates $C_x(a)$ in $\bG_b$. Therefore, no additional vertices can be reached from $a$ in $\bG_b$, yielding $C_x(a) = C_b(a)$, so Condition~\ref{it:block:1} holds.
\end{proof}

We now prove a quantitative lemma essential for characterizing \AT-free graphs. Let $\mC_x = \{C_{x}(v) \mid v \in \bG_x\}$ denote the set of neighborhood-deleted components of $\bG_x$. For any subset of components $\mH \subseteq \bigcup_{x \in V} \mC_{x}$, this lemma establishes, for every $x \in V$, a lower bound on the number of pairs $(y, C)$ such that $C$ is a member of both $\mC_x \cap \mH$ and $\mC_y \cap \mH$.

\begin{lemma}\label{lem:AT_deg_bound}
    Let $\bG = (V, E)$ be an \AT-free graph. Let $\mC_x = \{C_{x}(v) \mid v \in \bG_x\}$ be the set of connected components in $\bG_x$, and define $\mC \coloneqq \bigcup_{x \in V} \mC_{x}$. 

    For any subset $\mH \subseteq \mC$, let $\mH_x \coloneqq \mC_x \cap \mH$, let $s_x \coloneqq \max_{C \in \mH_x} |C|$ denote the size of the largest component in $\mH_x$, and let $V_x \coloneqq \bigcup_{C \in \mH_x} C$ be the set of all vertices in $\bG_x$ belonging to components in $\mH_x$.
    Then, for every $x \in V$ such that $\mH_x$ is non-empty, we have:
    \[
        |\{(y,C) \in V_x \times \mH_x \mid C \in \mH_y\}| \geq |V_x|^2 / (4s_x) - |V_x|.
    \]
\end{lemma}
\begin{proof}
    Fix $x \in V$.
    We begin by defining an auxiliary undirected graph $\bH_x$ on the vertex set $V_x$. An edge $\{a,b\}$ exists in $\bH_x$ if and only if $C_x(a) \neq C_x(b)$, $C_x(a) \neq C_{b}(a)$, and $C_x(b) \neq C_{a}(b)$.
    
    \begin{claim}
        The graph $\bH_x$ is triangle-free; otherwise, $\bG$ would contain an \AT.
    \end{claim}
    
    \begin{claimproof}
        Suppose $\bH_x$ contains a triangle on vertices $a,b,c$. We show these vertices form an \AT in $\bG$. Since all edges $\{a,b\}, \{b,c\}, \{c,a\}$ exist in $\bH_x$, we observe:
        \begin{enumerate}
            \item The vertices $a, b,$ and $c$ are pairwise non-adjacent. This is because from the existence of the edges $\{a,b\}, \{b,c\}, \{c,a\}$ in $\bH_x$ follows that $a$, $b$, and $c$ belong to distinct connected components of $\bG_x$.
            \item There exists an $a$-$b$ path in $\bG$ that avoids $N[c]$. Since $x \notin N[c]$ (as $c \in \bG_x$), it suffices to show that there are $a$-$x$ and $b$-$x$ paths avoiding $N(c)$. By symmetry, we focus on the $a$-$x$ path.
            
            Since $(a,c) \in \bH_x$, we have $C_x(a) \neq C_c(a)$. By Proposition~\ref{prp:block}, there exists a vertex $v \in N(x) \cap N(C_x(a))$ such that $v \notin N(c)$. Let $a' \in C_x(a)$ be a neighbor of $v$. An $a$-$x$ path avoiding $N(c)$ is formed by concatenating a path from $a$ to $a'$ within $C_x(a)$, the edge $\{a', v\}$, and the edge $\{v, x\}$.
        \end{enumerate}
        By symmetry, $b$-$c$ and $a$-$c$ paths exist that avoid $N[a]$ and $N[b]$, respectively. Thus, $a,b,c$ is an \AT, contradicting that $\bG$ is \AT-free.
    \end{claimproof}

    By Mantel's Theorem\footnote{Mantel's Theorem states that a triangle-free graph on $n$ vertices has at most $n^2/4$ edges.} \cite{M07}, the graph $\bH_x$ contains at most $|V_x|^2/4$ edges. This bound allows us to establish the following inequality.

    \begin{claim}
        The following inequality holds:
        $|\{(a,b) \in V_x \times V_x \mid C_{x}(a) = C_{b}(a)\}| \geq |V_x|^2 / 4 - |V_x| \cdot s_x$.
    \end{claim}
    
    \begin{claimproof}
        Recall that $s_x = \max_{C \in \mH_x} |C|$. We first note that the number of (unordered) pairs $\{a,b\}$, within the same component of $\bG_x$ is upper bounded by:
        \[
           |\{\{a,b\} \subseteq V_x \mid a \neq b, C_x(a) = C_x(b)\}| = \sum_{C \in \mH_x} \binom{|C|}{2} = \sum_{C \in \mH_x} \left( \frac{|C|^2}{2} - \frac{|C|}{2} \right) \leq  (s_x - 1) \cdot \sum_{C \in \mH_x} \frac{|C|}{2} = s_x \cdot \frac{|V_x|}{2} - \frac{|V_x|}{2}. 
        \]
        Let $m_x$ be the number of edges in $\bH_x$. By the definition of $\bH_x$, an edge $\{a,b\}$ is absent if and only if $C_x(a) = C_x(b)$, $C_x(a) = C_b(a)$, or $C_x(b) = C_a(b)$, i.e., if $\{a,b\}$ is in either of these sets:
        \[
            |\{\{a,b\} \subseteq V_x \mid a \neq b, C_x(a) = C_x(b)\}| \text{ or } |\{\{a,b\} \subseteq V_x \mid C_x(a) = C_b(a) \text{ or } C_x(b) = C_a(b)\}|.
        \]
        Note that these two sets are indeed disjoint: if $C_x(a) = C_x(b)$, then $b \in C_x(a)$. In this case, $C_x(a)$ cannot be a component in $\bG_b$ because the vertex $b$ itself is removed. Similar holds for $C_x(b)$ in $\bG_a$.
        This allows us to lower bound $m_x$ through the number of non-edges as:
        \begin{align*}
            m_x &= \binom{|V_x|}{2} -  |\{\{a,b\} \subseteq V_x \mid a \neq b, C_x(a) = C_x(b)\}|  - |\{\{a,b\} \subseteq V_x \mid C_x(a) = C_b(a) \text{ or } C_x(b) = C_a(b)\}| \\
            &\geq \frac{|V_x|^2}{2} - \frac{|V_x|}{2} - s_x \cdot \frac{|V_x|}{2} + \frac{|V_x|}{2} - |\{\{a,b\} \subseteq V_x \mid C_x(a) = C_b(a) \text{ or } C_x(b) = C_a(b)\}|.
        \end{align*}
        Simplifying the latter inequality, combining it with $m_x \leq |V_x|^2 / 4$, and rearranging, yields that $|\{\{a,b\} \subseteq V_x \mid C_x(a) = C_b(a) \text{ or } C_x(b) = C_a(b)\}| \geq |V_x|^2 / 4 - s_x |V_x|$. The claim follows by noticing that in $|\{(a,b) \in V_x \times V_x \mid C_{x}(a) = C_{b}(a)\}|$ each element $\{a,b\}$ of the latter set is counted at least once as $(a,b)$ or $(b,a)$.
    \end{claimproof}

    Finally, we establish the lower bound for $|\{(y,C) \in V_x \times \mH_x \mid C \in \mH_y\}|$. 
    For any pair $(a,b) \in V_x \times V_x$ such that $C_x(a) = C_b(a)$, the pair $(b, C_x(a))$ is a valid entry in the set $\{(y, C) \in V_x \times \mH_x \mid C \in \mH_y\}$. However, multiple vertices $a$ can belong to the same component $C = C_x(a)$. Specifically, for a fixed $b$ and $C$, there are exactly $|C|$ vertices $a$ such that $C_x(a) = C$. Since $|C| \leq s_x$ for all $C \in \mH_x$, each pair $(b, C)$ is counted at most $s_x$ times in the set of pairs $(a,b)$.
    Therefore, we have:
    \begin{align*}
        |\{(y,C) \in V_x \times \mH_x \mid C \in \mH_y\}| &\geq |\{(a,b) \in V_x \times V_x \mid C_{x}(a) = C_{b}(a)\}| / s_x
        \geq |V_x|^2 / (4s_x) - |V_x|. \qedhere
    \end{align*}
\end{proof}

This brings us to our main structural result for \AT-free graphs.

\atfreesmall

\begin{proof}
    Let $\mC \coloneqq \bigcup_{x \in V} \mC_{x}$.
    We begin by partitioning $\mC$ in $\Oh(\log^2 n)$ sets,
    depending on the size of $C \in \mC$ and in how many $\mC_x$ a component $C$ appears.
    More precisely, for $i,j \in \fragment{0}{\ceil{\log n}}$, we set
    \[
        \mC^{i,j} \coloneqq \{C \in \mC \mid 2^{i} \leq |C| < 2^{i+1}, 2^{j} \leq |\{x \mid C \in \mC_x\}| < 2^{j+1}\}\}
        \quad \text{and} \quad \mC_{x}^{i,j} \coloneqq \{C \in \mC^{i,j} \mid C \in \mC_x\}\text{ for $x \in V$.}
    \]

    Moreover, we introduce the auxiliary bipartite graph $\bF^{i,j}$ with vertices $V \cup \mC^{i,j}$ and an edge $(x,C)$ if $C \in \mC_{x}$. That is, $\mC^{i,j}_x$ is the neighborhood of $x$ in $\bF^{i,j}$. The graph $\bF^{i,j}$ helps us to bound $|\mC^{i,j}|$ as follows.
    
    \begin{claim}
        For every $i,j \in \fragment{0}{\ceil{\log n}}$, sucht that $\mC^{i,j} \neq \emptyset$, we have $|\mC^{i,j}| \leq (2^{-i+6}+2^{-j+5}) \cdot n$.
    \end{claim}
    \begin{claimproof}
        Fix such $i,j$, and let $\deg(x)$ and $\deg(C)$ denote the degree of $x$ and $C$ in $\bF^{i,j}$, respectively. 

        For every $x \in V$, let $V_x \coloneqq \bigcup_{C \in \mC^{i,j}_x} C$ and $D_x \coloneqq |\{(y,C) \in V_x \times \mC^{i,j}_x \mid C \in \mC^{i,j}_y\}|$. If $\deg(x)>0$, we apply \cref{lem:AT_deg_bound} to $\bG$ with $\mH=\mC^{i,j}$ and obtain
        \begin{center}
            $D_x \geq |V_x|^2/(4s_x) - |V_x|$, where $s_x \coloneqq \max_{C \in \mC^{i,j}_x} |C|$.
        \end{center}
        The number $D_x$ counts a subset of the walks of length two from vertex $x$ to some vertex $y$ through a component $C \in \mC^{i,j}_x \cap \mC^{i,j}_y$ in the graph $\bF^{i,j}$.
        Since every $C \in \mC_{x}^{i,j}$ satisfies $2^{i} \leq |C| < 2^{i+1}$, for $\deg(x)>0$ we have:
        \begin{align}
            D_{x}
            \geq |V_x|^2/ (4s_x) - |V_x|
            \geq (\deg(x) \cdot 2^{i})^2 / (4 \cdot 2^{i+1}) -  \deg(x) \cdot 2^{i+1}
            = \deg(x)^2 \cdot 2^{i-3} - \deg(x) \cdot 2^{i+1}.
            \label{eq:bound_cc_at:0} 
        \end{align}
        If $\deg(x)=0$, then $V_x=\emptyset$ and $D_x=0$, so \cref{eq:bound_cc_at:0} holds in this case as well.
        We can use this last inequality in the following chain of inequalities:
        \begin{align}
            2 \cdot |\mC^{i,j}| \cdot 2^{2j+2} 
            &> 2 \cdot \textstyle\sum_{C \in \mC^{i,j}} \deg(C)^{2}
            \label{eq:bound_cc_at:1} \\
            &\geq \textstyle\sum_{x \in V} D_x 
            \label{eq:bound_cc_at:2} \\
            &\geq \textstyle\sum_{x \in V} \left( \deg(x)^2 \cdot 2^{i-3} - \deg(x) \cdot 2^{i+1}\right) 
            \label{eq:bound_cc_at:3} \\
            &\geq 2^{i-3} /n \cdot \big( \textstyle\sum_{C \in \mC^{i,j}} \deg(C) \big)^2 - 2^{i+1} \cdot \textstyle\sum_{C \in \mC^{i,j}} \deg(C) 
            \label{eq:bound_cc_at:4} \\
            &> 2^{i-3+2j}/n \cdot |\mC^{i,j}|^2 - 2^{i+j+2} \cdot |\mC^{i,j}|, \label{eq:bound_cc_at:5}
        \end{align}
        where we use:
        \begin{description}
            \item[\cref{eq:bound_cc_at:1}:] By the definition of $\mC^{i,j}$, we have $\deg(C) < 2^{j+1}$ for any $C \in \mC^{i,j}$.
            \item[\cref{eq:bound_cc_at:2}:] The expression $\sum_{C \in \mC^{i,j}} \deg(C)^2$ counts all length-two walks in $\bF^{i,j}$ whose endpoints lie in $V$ and whose middle vertex lies in $\mC^{i,j}$, while $\sum_{x \in V} D_x$ counts a subset of these walks.
            \item[\cref{eq:bound_cc_at:3}:] \cref{eq:bound_cc_at:0}.
            \item[\cref{eq:bound_cc_at:4}:] Cauchy–Schwarz inequality on a length-$n$ vector that contains $\deg(x)$ in the entry indexed by $x \in V$ and an all-one vector of the same length. We also use that $\sum_{x \in V} \deg(x) = \sum_{C \in \mC^{i,j}} \deg(C)$.
            \item[\cref{eq:bound_cc_at:5}:] By the definition of $\mC^{i,j}$, we have $2^{j} \leq \deg(C) < 2^{j+1}$ for any $C \in \mC^{i,j}$.
        \end{description}
        
        The previous chain of inequalities yields immediately the claim:
        \[
            |\mC^{i,j}| < (2^{2j+2}+2^{i+1+j}) \cdot n / 2^{i-3+2j} = (2^{-i+6}+2^{-j+5}) \cdot n. \claimqedhere
        \]
    \end{claimproof}

    Since $\{\mC^{i,j}\}_{i,j}$ partitions $\mC$, we conclude that:
    \begin{align*}
        |\mC| &\leq \sum_{i=0}^{\lceil\log n\rceil} \sum_{j=0}^{\lceil\log n\rceil} |\mC^{i,j}| \leq n \cdot \Big( \sum_{i=0}^{\lceil\log n\rceil} \sum_{j=0}^{\lceil\log n\rceil} 2^{-i+6} + \sum_{j=0}^{\lceil\log n\rceil} \sum_{i=0}^{\lceil\log n\rceil} 2^{-j+5} \Big) \leq n \cdot \lceil\log n\rceil \cdot (2^{7}+2^{6}) \leq 400 n\log n, 
    \end{align*} 
    as desired.
\end{proof}

Finally, we can show how to determine in time $\Ohtilde(n^{\omega})$ whether a graph is \AT-free by exploiting the structural result on \AT-free graphs that we just established.

\algatfree

\begin{proof}
    Given an input graph $\bG = (V, E)$, our algorithm proceeds as follows:
    \begin{enumerate}[(i)]
        \item First, we apply \cref{lem:ndcc_to_bmm} to identify the collection $\mC$ of all neighborhood-deleted components. Formally, let $\mC_x = \{C_{x}(v) \mid v \in V \setminus N[x]\}$ denote the set of connected components in $\bG_x$. Then, $\mC = \bigcup_{x \in V} \mC_{x}$. 
        \label{it:alg_atfree:i}
        \item Via \cref{lem:fingerprints}, we assign each component $C \in \mC$ a unique identifier $h(C)$, which allows us to count the number of distinct components in $\Ohtilde(n^2)$ time. If $|\mC| > 400n \log n$, we terminate the process and report that $\bG$ is not \AT-free.
        \label{it:alg_atfree:ii}
        \item Otherwise, for each distinct $C \in \mC$, we fix an arbitrary \emph{representative} vertex $r_C \in V$ such that $C \in \mC_{r_C}$. We then construct a triangle detection instance on a graph $\bT$ with $\Oh(n \log n)$ vertices defined as follows:
        \begin{itemize}
            \item The vertex set of $\bT$ is $\mC$.
            \item For every non-edge $\{a, b\} \notin E$, we add an edge between $C_{a}(b)$ and $C_{b}(a)$ in $\bT$ if and only if $r_{C_{a}(b)} = a$ and $r_{C_{b}(a)} = b$.
        \end{itemize}
        Finally, we report that $\bG$ is \AT-free if and only if $\bT$ is triangle-free.
        \label{it:alg_atfree:iii}
    \end{enumerate}

    \subparagraph*{Correctness.} 
    If the algorithm terminates in Step~\ref{it:alg_atfree:ii}, correctness follows directly from \cref{thm:atfree_small}.
    If this is not the case, correctness follows from this claim.

    \begin{claim}
        $\bG$ is \AT-free if and only if $\bT$ is triangle-free.
    \end{claim}
    \begin{claimproof}
        To prove the forward direction, suppose $\bT$ contains a triangle with vertices $\{A, B, C\}$. We claim that their representatives $r_A, r_B, r_C$ form an \AT in $\bG$.
        Since $\{A, B\}$, $\{B, C\}$, and $\{C, A\}$ are edges in $\bT$, it follows by construction that $\{r_A, r_B\}$, $\{r_B, r_C\}$, and $\{r_C, r_A\}$ are non-edges in $\bG$. Furthermore, since $\{A, B\}$ and $\{C, A\}$ are edges in $\bT$, we must have $C_{r_A}(r_B) = C_{r_A}(r_C) = A$. This implies the existence of an $r_B$-$r_C$ path in $\bG$ that avoids the neighborhood $N[r_A]$. By symmetry, there exist $r_A$-$r_C$ and $r_A$-$r_B$ paths that avoid $N[r_B]$ and $N[r_C]$, respectively.

        Conversely, assume that $\{a, b, c\}$ is an \AT in $\bG$. 
        Let $A \coloneqq C_a(b) = C_a(c)$, $B \coloneqq C_b(a) = C_b(c)$, and $C \coloneqq C_c(a) = C_c(b)$. Note that if $\{r_A, r_B, r_C\} = \{a, b, c\}$, then $r_A, r_B, r_C$ is an \AT: the previously established argument for the forward direction applies conversely to show that $A, B, C$ forms a triangle in $\bT$.

        Thus, it suffices to show that $r_A, b, c$ remains an \AT in $\bG$. Note that $r_A$ must be distinct from $b,c$, since $b,c$ are vertices in $\bG_{r_A}$. If we can demonstrate that the property of being an \AT is preserved when replacing $a$ with the corresponding representative $r_A$, then by symmetry, replacing $b$ with $r_B$ ensures that $r_A, r_B, c$ is an \AT. Applying this logic once more implies that $r_A, r_B, r_C$ is an \AT, and thus we get the existence of the triangle $A,B,C$ in $\bT$.

        Assume $r_A \neq a$, as the case $r_A = a$ is trivial. 
        Since $A \in \mC_a$ and $A \in \mC_{r_A}$, we observe that $N(r_A) \cap A = \emptyset$; otherwise, $A$ could not be a connected component in $\bG_{r_A}$. Because $b, c \in A$, this immediately implies that $r_A$ is non-adjacent to both $b$ and $c$.
        
        Furthermore, since $A$ is a component in $\bG \setminus N[r_A]$, there exists a $b$-$c$ path that avoids $N[r_A]$. To prove that $\{r_A, b, c\}$ is an \AT, it remains to show that there exists a $b$-$r_A$ path avoiding $N[c]$ and, by symmetry, a $c$-$r_A$ path avoiding $N[b]$.
        As $a, b, c$ is an \AT, there already exists a $b$-$a$ path $\pi$ that avoids $N[c]$. Let $a'$ be the neighbor of $a$ on $\pi$. Without loss of generality, we can assume $a' \in N(a) \cap N(A)$. We now consider two cases for the construction of the $b$-$r_A$ path:
        \begin{enumerate}
            \item If $r_A \in N(a)$, we concatenate $\pi$ with the edge $\{a, r_A\}$. Since $r_A \notin N(c)$, this path avoids $N[c]$.
            
            \item If $r_A \notin N(a)$, we note that $r_A \notin A$ (otherwise $A$ would not be a component in $\bG \setminus N[r_A]$). By \cref{prp:block} applied to $a$, $r_A$, and any $v \in A$, it follows that every vertex $v \in N(a) \cap N(A)$ must also be a neighbor of $r_A$. In particular, $a' \in N(r_A)$. We can therefore construct a valid path by concatenating the $b$-$a'$ subpath of $\pi$ (which avoids $N[c]$) with the edge $\{a', r_A\}$. \claimqedhere
        \end{enumerate}
    \end{claimproof}

    \subparagraph*{Runtime.}
    Step~\ref{it:alg_atfree:i} and Step~\ref{it:alg_atfree:iii} take time $\Ohtilde(n^{\omega})$. On the other hand, Step~\ref{it:alg_atfree:ii} only takes time $\Ohtilde(n^{2})$.
\end{proof}

\subsection{Listing Dominating Pairs in Matrix Multiplication Time}

\dompairs

\begin{proof}
   Let $\bG = (V,E)$ be an undirected graph.

    We begin by observing that we may assume each $\bG_v$ contains at least one node. Indeed, if this were not the case, then $N[v] = V$. In such an instance, $v$ can be removed from $\bG$ because it is already known that $v$ is adjacent to every vertex (and path) in $\bG$. Moreover, in this case we know that $u, v$ constitutes a dominating pair for all $u \in V \setminus \{v\}$.

    Now, for each pair of distinct nodes $u,v \in V$, we must determine whether $u,v$ constitutes a dominating pair. To this end, we observe that $u,v$ is a dominating pair if and only if for each node $w \notin N[u] \cup N[v]$, we have $C_w(u) \neq C_w(v)$. This condition follows because $C_w(u) \neq C_w(v)$ indicates that there is no $u$-$v$ path that avoids the neighborhood of $w$. 
    Conversely, for any node $w \in N[u] \cup N[v]$, every $u$-$v$ path trivially intersects the closed neighborhood of $w$, as either $u$ or $v$ lies in $N[w]$; and for any node $w \notin N[u] \cup N[v]$, if we have $C_w(u) \neq C_w(v)$, then any $u$-$v$ path must intersect $N(w)$.

    \subparagraph*{Our algorithm.}
    Our algorithm, rather than listing all dominating pairs, will list all non-dominating pairs, by trying to find for each $u,v$ a \emph{witness} node $w \notin N[u] \cup N[v]$ such that $C_w(u) = C_w(v)$. Specifically, the algorithm proceeds as follows:
    \begin{enumerate}[(i)]
        \item First, we use \cref{lem:ndcc_to_bmm} to find all neighborhood-deleted components. For each $v \in V$, let $\mC_v$ be the connected components in $\bG_v$ and set $\mC \coloneqq \bigcup_{v \in V} \mC_v$. Similarly to \cref{them:alg_atfree}, we use \cref{lem:fingerprints}, to compute the count $|\{v \mid C \in \mC_v\}|$ of the number of times each connected component $C$ appears across the graphs $\bG_v$.
        \label{alg:dompairs:i}
        \item For each $v \in V$, let $M_v$ denote the largest connected component in $\mC_v$, i.e., $M_v \coloneqq \argmax_{C \in \mC_v} |C|$. If multiple components share the same maximum size, we fix one arbitrarily. We then construct a Boolean $n \times n$ matrix $A$, indexed by the vertices in $V$, defined as follows:
        \[
            A[u,w] = 
            \begin{cases} 
                1 & \text{if } u \in M_w \\
                0 & \text{otherwise} 
            \end{cases}
        \]
        By computing the BMM product $A A^{\top}$, we identify and mark as non-dominating all pairs $u,v$ for which $(A A^{\top})[u,v] = 1$. That is we mark all pairs that are witnessed by a component in $\{M_v\}_{v \in V}$.
        \label{alg:dompairs:ii}
        \item Next, for each $C \in \mC$, we do the following. We let $x_C \in V$ be the node with $C \in \mC_{x_C}$ that minimizes the number of nodes in $\bG_{x_C}$, i.e., 
        \[
            x_C \coloneqq \argmin_{w \in V \mid C \in \mC_w} |V(\bG_w)| = \argmax_{w \in V \mid C \in \mC_w} |N[w]|.
        \]
        If $C$ appears in less than $|V(\bG_{x_C})|/2 = (n - |N[x_C]|)/2$ graphs $\bG_w$, i.e., if $|\{w \in V \mid C \in \mC_w\}| < (n - |N[x_C]|)/2$, we do nothing. Otherwise, if $|\{w \in V \mid C \in \mC_w\}| \geq (n - |N[x_C]|)/2$ we iterate over all pairs of nodes $u,v \in C$ and mark them as non-dominating pairs.
        \label{alg:dompairs:iii}
        \item Return as dominating pairs all pairs $u,v$ that have not been marked as non-dominating in either Step~\ref{alg:dompairs:ii} or Step~\ref{alg:dompairs:iii}. 
        \label{alg:dompairs:iv}
    \end{enumerate}

    \subparagraph*{Correctness.}
    We show that for each non-dominating pair $u,v$, we identify at least one witness $w$ and mark $u,v$ as non-dominating, either in Step~\ref{alg:dompairs:ii} or in Step~\ref{alg:dompairs:iii}.
    
    To this end fix a non-dominating pair $u, v$. 
    For each witness $w$ of $u,v$, let $C_w(u)=C_w(v)$ be the component containing both vertices. Choose $w^*$ among these witnesses to maximize $|C_w(u)|$, and set $C \coloneqq C_{w^*}(u)=C_{w^*}(v)$.
    We begin by observing that $(A A^{\top})[u,v] = 1$ if $C = M_{w^*}$. Consequently, if $C = M_{w^*}$, the pair $u,v$ is marked as non-dominating in Step~\ref{alg:dompairs:ii} and we may henceforth assume that $C \neq M_{w^*}$.

    The following claim establishes that $u,v$ is marked in Step~\ref{alg:dompairs:iii}.

    \begin{claim}\label{clm:dompairs}
        The number of $w \in V$ such that $C \in \mC_w$ is at least $|V(\bG_{w^*})|/2 = (n - |N[w^*]|)/2$, i.e., $|\{w \in V \mid C \in \mC_w\}| \geq (n - |N[w^*]|)/2$.
        In particular, because $n - |N[w^*]| \geq n - |N[x_C]|$ we have that $u,v$ is marked as a non-dominating pair in Step~\ref{alg:dompairs:iii}.
    \end{claim}
    \begin{claimproof}
        Given our assumption that $C \neq M_{w^*}$, we know that $|C| \leq (n - |N[w^*]|) / 2$ (otherwise, $C$ would be $M_{w^*}$). Consequently, the number of nodes in $\bG_{w^*}$ that are not contained in $C$ is at least 
        \[
            |\{w \in \bG_{w^*} \mid w \notin C\}| = (n - |N[w^*]|) - |C| \geq (n - |N[w^*]|) / 2.
        \]
        To conclude, it suffices to show that for each $w \in V \setminus (N[w^*] \cup C)$ we have $C \in \mC_w$.

        To this end, let $w \in \bG_{w^*}$ be an arbitrary node with $w \notin C$. Since $w \notin C$, we have that $N[w] \cap C = \emptyset$. Because of this, each pair of nodes in $C$ remains connected in $\bG_w$. This implies that $w$ witnesses $u,v$ and $C' \coloneqq C_w(u) = C_w(v)$ satisfies $C \subseteq C'$.
        However, by our choice of $w^*$, we must have that $|C'| \leq |C|$, which means that $C' = C$. We conclude that $C \in \mC_w$, as desired.
    \end{claimproof}

    \subparagraph*{Runtime.}
    Clearly, Step~\ref{alg:dompairs:i} and Step~\ref{alg:dompairs:ii} run in time $\Ohtilde(n^{\omega})$. 
    
    On the other hand, in Step~\ref{alg:dompairs:iii}, we first find $x_C$ for each $C \in \mC$ and verify whether to list all witnessed pairs in $C$. It is not difficult to see that this can be implemented in time $\Oh(\sum_{v \in V} |\mC_v|) \leq \Oh(\sum_{v \in V} |V \setminus N[v]|) \leq \Oh(n^2)$.
    Let $S \subseteq \mC$ be the set of all $C \in \mC$ for which we list all pairs $u,v \in C$. For each $C \in S$ we spend time $\Oh(|C|^2)$. Over all elements $C \in S$ this becomes:
    \begin{align*}
         \sum_{C \in S} |C|^2 &\leq \sum_{C \in S} |C| \cdot |V \setminus N[x_C]| && (C \in \mC_{x_C})\\
         &\leq 2 \sum_{C \in S} |C| \cdot |\{w \in V \mid C \in \mC_w\}| && (|\{w \in V \mid C \in \mC_w\}| \geq |V \setminus N[x_C]| / 2 \text{ holds for all } C \in S) \\
         &\leq 2 \sum_{C \in \mC} |C| \cdot |\{w \in V \mid C \in \mC_w\}| && (S \subseteq \mC) \\
         &= 2 \sum_{v \in V} |V \setminus N[v]| && (\text{double counting argument}) \\
         &\leq 2n^2,
    \end{align*}
    which concludes the proof.
\end{proof}

\subsection{Listing Two-Pairs in Matrix Multiplication Time}

We begin with the observation by Arikati and Pandu Rangan \cite{AP1991}, which characterizes two-pairs based on neighborhood-deleted components. For the sake of completeness, we also include a proof of it.

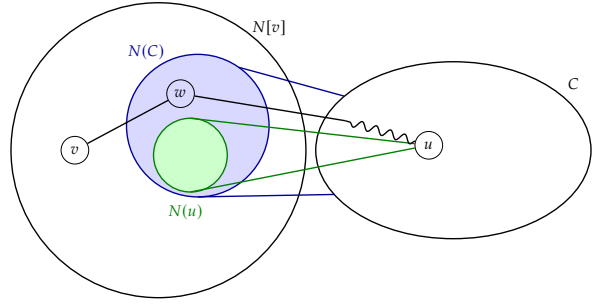
\begin{figure}[t]
    \centering
    \scalebox{0.65}{\usetikzlibrary{decorations.pathmorphing,calc}

\begin{tikzpicture}[
    scale=1,
    every node/.style={font=\small},
    vertex/.style={circle, draw, inner sep=1.2pt, minimum size=16pt},
    pathstyle/.style={decorate, decoration={snake, amplitude=0.8mm, segment length=3mm}}
]

% Outer neighbourhood N[v]
\draw[thick] (0,0) circle (3);
\node at (2.25,2.5) {$N[v]$};

% Vertex v
\node[vertex] (v) at (-1.7,0) {$v$};

% Region N(C)
\def\NCx{0.8}
\def\NCy{0.5}
\def\NCr{1.45}
\fill[blue!15] (\NCx,\NCy) circle (\NCr);
\draw[blue!60!black, thick] (\NCx,\NCy) circle (\NCr);
\node[blue!60!black] at (-0.25,2) {$N(C)$};

% Region N(u)
\def\NUx{0.65}
\def\NUy{-0.1}
\def\NUr{0.75}
\fill[green!20] (\NUx,\NUy) circle (\NUr);
\draw[green!50!black, thick] (\NUx,\NUy) circle (\NUr);
\node[green!50!black] at (0.55,-1.25) {$N(u)$};

% Vertex w in N(C)\N(u)
\node[vertex] (w) at (0.45,1.15) {$w$};

% Edge v-w
\draw[thick] (v) -- (w);

% Component C
\draw[thick] (6,0) ellipse (2.8 and 1.8);
\node at (8.45,1.35) {$C$};

% Vertex u inside C
\node[vertex] (u) at (5.5,0.1) {$u$};

% Boundary points of C
\coordinate (Ctop) at (3.78,1.10);
\coordinate (Cbot) at (3.58,-0.90);

% Attachment points on N(C)
% Top point chosen on the border of the N(C) circle, a bit lower/right
\coordinate (NCtop) at ($(\NCx,\NCy)+(55:\NCr)$);
\coordinate (NCbot) at (\NCx,\NCy-\NCr);

% Top and bottom of N(u)
\coordinate (NUtop) at (\NUx,\NUy+\NUr);
\coordinate (NUbot) at (\NUx,\NUy-\NUr);

% Blue connector lines: C to N(C)
\draw[blue!60!black, thick] (Ctop) -- (NCtop);
\draw[blue!60!black, thick] (Cbot) -- (NCbot);

% Green connector lines: u to N(u)
\draw[green!50!black, thick] (u) -- (NUtop);
\draw[green!50!black, thick] (u) -- (NUbot);

% Wiggly path inside C, then one final straight edge to w
\coordinate (p3) at (3.8,0.6); % point on the boundary of C

\draw[pathstyle, thick] (u) -- (p3);
\draw[thick] (p3) -- (w);

\end{tikzpicture}}
    \caption{Criterion for two vertices to form a two-pair}
    \label{fig:criterion}
\end{figure}

\begin{proposition}[Lemma 2 of \cite{AP1991}]\label{prp:char_twopair}
    Let $\bG = (V,E)$ be a graph. Then, for any two distinct non-neighboring vertices $u,v$ the following holds: $u,v$ is a two-pair if and only if $N(u) \cap N(v) = N(C_v(u))$.
\end{proposition}
\begin{proof}
    Let us abbreviate $C \coloneqq C_v(u)$. By definition of $C$, we always have that $N(u) \cap N(v) \subseteq N(C)$ and $N(C) \subseteq N(v)$. 
    
    Suppose first that $N(u) \cap N(v) \neq N(C)$, or equivalently that $N(C)  \setminus N(u) \neq \emptyset$.
    Let $w$ be any vertex in $N(C) \setminus N(u)$.
    Consider a shortest path $\pi$ from $u$ to $w$ in $\bG[C \cup \{w\}]$. Then, extending the path $\pi$ with the edge $\{w,v\}$ gives a chordless path from $u$ to $v$, which has length at least three, as illustrated in Figure \ref{fig:criterion}. Thus, $u,v$ are not a two-pair.

    Conversely, suppose that $u,v$ are not a two-pair, which means that there is a chordless path of length at least three between $u$ and $v$. Then, the neighbor $w$ of $v$ on this path is in $N(C)$ but not in $N(u)$. Thus, we have that $N(u) \cap N(v) \neq N(C)$.
\end{proof}

In other words, a two-pair is characterized by the property that the entire out-neighborhood of $C_v(u)$ (which must be a subset of $N(v)$) coincides exactly with the neighbors of $u$ outside of $C_v(u)$ (which must also be a subset of $N(v)$).
This allows us to give the following algorithm to list all two-pairs in matrix multiplication time.

\twopairs

\begin{proof}
    Let $\bG = (V,E)$ be an undirected graph.

    \subparagraph{Our algorithm.} We proceed as follows:
    \begin{enumerate}[(i)]
        \item We first precompute several quantities. Using \cref{lem:ndcc_to_bmm}, we find all neighborhood-deleted components. Furthermore, via matrix multiplication, we compute for each pair of vertices $u, v$ the number of common neighbors $|N(u) \cap N(v)|$ and the size of their set difference $|N(u) \setminus N(v)|$.
        We also set up the fingerprints of \cref{lem:fingerprints}.
        \label{it:twopairs:i}
        
        \item Next, for each node $v \in V$, we consider $\bG_v$ and the set of connected components $\mC_v$ within it. For each $C \in \mC_v$, we iterate over all nodes $w \in C$ to compute the set $M_v(C) \coloneqq \{w \mid |N(v) \cap N(w)| = \max_{x \in C} |N(v) \cap N(x)|\}$, which consists of all nodes in $C$ that maximize the number of common neighbors with $v$ in $\bG$. From each set $M_v(C)$, we also select an arbitrary \emph{representative} node $r_v(C)$.
        
        Among all nodes $w \in C$, we restrict our testing to whether $v$ and the representative $r_v(C)$ form a two-pair. By \cref{prp:char_twopair}, any node in $C$ that does not maximize the number of common neighbors with $v$ in $\bG$ cannot form a two-pair with $v$. Furthermore, it is easy to observe that either all nodes in $M_v(C)$ form a two-pair with $v$, or none of them do.
        \label{it:twopairs:ii}
        
        \item For each node $v \in V$ and each $C \in \mC_v$, we iterate over all nodes $w \in C \setminus N[r_{v}(C)]$ in $\Oh(|C|)$ time. If there exists a node $w$ satisfying $C_{r_{v}(C)}(w) = C_{r_{v}(C)}(v)$, we mark $v$ and $r_v(C)$ as not a two-pair.
        \label{it:twopairs:iii}
    
        \item For each node $v \in V$ and each $C \in \mC_v$, we again iterate over all nodes $w \in C$ to compute the following quantity:
        \begin{align}
        \Delta_v(C) \coloneqq \quad \smashoperator{\sum_{w \in C \setminus N[r_v(C)]}} \quad |N(w) \cap N(v)| 
        - \quad \smashoperator{\sum_{w \in C \setminus N[r_v(C)]}} \quad |N(w) \cap N(r_v(C))| 
        + \quad \smashoperator{\sum_{w \in C \cap N(r_v(C))}} \quad |N(w) \setminus N[r_v(C)]|.
        \label{eq:twopairs:0}
    \end{align}
        If $\Delta_v(C) > 0$, we mark $v$ and $r_v(C)$ as not a two-pair.
        \label{it:twopairs:iv}
        
        \item Finally, for each node $v \in V$, we list all nodes $u$ that form a two-pair with $v$. Specifically, for each $C \in \mC_v$, if $v$ and $r_v(C)$ have not been marked as not a two-pair, we report all nodes $u \in M_v(C)$ as forming a two-pair with $v$.

        \label{it:twopairs:v}
    \end{enumerate}

    \subparagraph*{Correctness.} 
    Let $v \in V$ and $C \in \mC_v$ be arbitrary.
    Moreover, abbreviate $u \coloneqq r_v(C)$.
    As elaborated earlier, it suffices to show that $v$ and $u$ are not marked if and only if $v$ and $u$ are a two-pair.

    We first observe that $u$ and $v$ are not a two-pair if and only if there is $w \in C \setminus \{u\}$ that satisfies $(N(w) \cap N(v)) \setminus N(u) \neq \emptyset$.
    We first prove that $u$ and $v$ are marked in Step~\ref{it:twopairs:iii} if and only if there is $w \in C \setminus N[u]$ with $(N(w) \cap N(v)) \setminus N(u) \neq \emptyset$.

    \begin{claim}
        There is a node $w' \in C \setminus N[u]$ with $C_{u}(w') = C_{u}(v)$
        if and only if there is a node $w \in C \setminus N[u]$ with $(N(w) \cap N(v)) \setminus N(u) \neq \emptyset$.
    \end{claim}
    \begin{claimproof}
        First, suppose there exists a node $w' \in C \setminus N[u]$ such that $C_{u}(w') = C_{u}(v)$. Consider an arbitrary path $\pi$ from $w'$ to $v$ in $\bG_u$. Let $x$ be the first node on $\pi$ that is not in $C$, and let $w \in C$ be its immediate predecessor on the path. Note that any neighbor of $w$ that is not in $C$ must belong to $N(v)$; therefore, $x \in N(v)$. Therefore, $w \in C \setminus N[u]$ and $(N(w) \cap N(v)) \setminus N(u) \neq \emptyset$. 

        Conversely, observe that $v \notin N(u)$, $w \notin N(u)$, and $(N(w) \cap N(v)) \setminus N(u) \neq \emptyset$ immediately imply that $v$ and $w$ are connected by a path of length two in $\bG_u$. Thus, $C_{u}(w) = C_{u}(v)$.
    \end{claimproof}

    It remains to prove that if there is no $w \in C \setminus N[u]$ with $(N(w) \cap N(v)) \setminus N(u) \neq \emptyset$, then $\Delta_v(C) > 0$ if and only if $u$ and $v$ are not a two-pair.

    \begin{claim}
        If there is no $w \in C \setminus N[u]$ with $(N(w) \cap N(v)) \setminus N(u) \neq \emptyset$, then $\Delta_v(C) = \sum_{w \in C \cap N(u)} |(N(w) \cap N(v)) \setminus N(u)|$. In particular, $\Delta_v(C) > 0$ if and only if there is $w \in C \setminus \{u\}$ with $(N(w) \cap N(v)) \setminus N(u) \neq \emptyset$.
    \end{claim}
    \begin{claimproof}
        We proceed by rewriting each of the three terms in \cref{eq:twopairs:0}. 
        For the first term, we use the fact that if $(N(w) \cap N(v)) \setminus N(u) = \emptyset$ for every $w \in C \setminus N[u]$, then:
        \begin{align*}
            \textstyle\sum_{w \in C \setminus N[u]} |N(w) \cap N(v)| = 
            \textstyle\sum_{w \in C \setminus N[u]} |N(w) \cap N(v) \cap N(u)|.
        \end{align*}
        For the second and third terms, basic set rewriting yields:
        \begin{align*}
            \textstyle\sum_{w \in C \setminus N[u]} |N(w) \cap N(u)| = 
            \textstyle\sum_{w \in C \setminus N[u]} |N(w) \cap N(u) \cap N(v)| + 
            \underbrace{\textstyle\sum_{w \in C \setminus N[u]} |(N(w) \cap N(u)) \setminus N(v)|}_{\symlabel{(\dagger)}{term:dagger}}, \\
            \textstyle\sum_{w \in C \cap N(u)} |N(w) \setminus N[u]|
            = \textstyle\sum_{w \in C \cap N(u)} |(N(w) \cap N(v)) \setminus N[u]|
            + \underbrace{\textstyle\sum_{w \in C \cap N(u)} |N(w) \setminus (N(v) \cup N[u])|}_{\symlabel{(\ddagger)}{term:ddagger}}.
        \end{align*}
        Summing and subtracting the terms as they appear in the definition of $\Delta_v(C)$, we get that the right-hand side simplifies to our target expression, $\sum_{w \in C \cap N(u)} |(N(w) \cap N(v)) \setminus N(u)|$, since $(N(w) \cap N(v)) \setminus N[u] = (N(w) \cap N(v)) \setminus N(u)$, and provided that Term~\ref{term:dagger} and Term~\ref{term:ddagger} are equal. To verify this equality, we note that:
        \begin{itemize}
            \item Term~\ref{term:ddagger} counts all pairs $(x,y)$ such that $x \in C \cap N(u)$, $\{x,y\} \in E$, and $y \notin N(v) \cup N[u]$.
            \item Term~\ref{term:dagger} counts all pairs $(x,y)$ such that $y \in C \setminus N[u]$, $\{x,y\} \in E$, and $x \in N(u) \setminus N(v)$.
        \end{itemize}
        Because every neighbor of $u$ is either in $C$ or in $N(v)$, the condition $x \in C \cap N(u)$ is equivalent to $x \in N(u) \setminus N(v)$. Moreover, for any node $y \in N(x)$ where $x \in C$, we have $y \notin N(v)$ if and only if $y \in C$. Thus, we have that Term~\ref{term:dagger} and Term~\ref{term:ddagger} are indeed equal.
    \end{claimproof}

    \subparagraph*{Runtime.} Step~\ref{it:twopairs:i} takes time $\Ohtilde(n^{\omega})$. All other steps take time $\Oh(\sum_{v} \sum_{C \in \mC_v} |C|) = \Oh(\sum_{v} n) = \Oh(n^2)$.
\end{proof}

\subsection{Recognizing Quasi-Triangulated Graphs Faster}

As explained in \cref{sec:to}, it suffices to give an algorithm for the decremental simplicial vertex problem, where each update corresponds to a vertex deletion and we need to maintain whether there is a simplicial vertex in the graph. In the following theorem, we present an algorithm that solves this problem in subcubic (amortized) time.

\decrementalsimplicial

\begin{proof}
    Let $\Delta \in \fragment{1}{n}$ be a parameter.
    First, we set up some notation and make some remarks:
    \begin{itemize}
        \item For $i \in \fragmentco{0}{n}$, let the $i$-th \emph{time step} be the point in time before the $i$-th removal (zero-indexed) in the overall process.
        Further partition the $n$ time steps into $\lceil n/\Delta \rceil$ \emph{phases} of size at most $\Delta$. Let $\fragmentco{t_j}{t_{j+1}}$ denote the time steps in the $j$-th phase, where $t_0 = 0$ and $t_{\lceil n/\Delta \rceil} = n$. 
        
        \item Let $\bG_i = (V_i, E_i)$ be the graph at the $i$-th time step (in particular $\bG_0$ is the input graph), write $N_i(v)$ for $N_{\bG_i}(v)$ and $N_i[v]$ for $N_i(v)\cup\{v\}$, and let $A_i$ and $\bar{A}_i$ denote the adjacency matrix of $\bG_i$ and the complement of $\bG_i$, respectively. Note that during the entire process it is easy to maintain $A_i$ and the neighborhood of the vertices in total time $\Oh(n^2)$.

        \item The algorithm maintains a count $\delta_i(v)$ for each vertex $v \in V_i$, defined as the number of induced $P_3$'s where $v$ is the middle vertex. Formally, $\delta_i(v) \coloneqq |\{\{x,y\} \subseteq N_{i}(v) \mid \{x,y\} \notin E_i\}|$ is the number of missing edges in the neighborhood of $v$. Clearly, $\delta_i(v) = 0$ if and only if $v$ is a simplicial vertex in $\bG_i$. This means that for each query, it suffices to return any node for which $\delta_i(v) = 0$, or report that no such node exists.
    \end{itemize}

    \subparagraph{Our Algorithm.}
    To maintain the values $\delta_i(v)$ up to date, suppose we just removed a vertex, and we need to compute $\delta_i$ for time step $i$. Then: 

    \begin{enumerate}[(i)]
        \item If $i = t_j$, i.e., we are at the start of phase $j$, then we compute
        \[
            \delta_i(v) = \frac{1}{2} \sum_{x \in N_i(v)} |N_i(v) \setminus N_i[x]|,
        \]
        where the values $|N_i(v) \setminus N_i[x]|$ for all pairs of vertices $v, x \in V_i$ are computed as follows:
        \begin{itemize}
            \item If $j = 0$ (and so $i = 0$), we compute these values via the product $A_0 \bar{A}_0$.
            \item Otherwise, if $j > 0$, let $D_j = V_{t_{j-1}} \setminus V_{t_j}$ be the set of vertices removed during the $(j-1)$-th phase. We obtain $|N_{t_j}(v) \setminus N_{t_j}[x]|$ from $|N_{t_{j-1}}(v) \setminus N_{t_{j-1}}[x]|$ by computing
            \[
                |N_{t_j}(v) \setminus N_{t_j}[x]| = |N_{t_{j-1}}(v) \setminus N_{t_{j-1}}[x]| - |(N_{t_{j-1}}(v) \setminus N_{t_{j-1}}[x]) \cap D_j|.
            \]
            We observe that we can compute the second term simultaneously for all $v, x \in V_{t_j}$ by multiplying the $n \times |D_j|$ adjacency matrix between $V_{t_{j-1}}$ and $D_j$ with the $|D_j| \times n$ non-adjacency matrix between $D_j$ and $V_{t_{j-1}}$, whose entry for $(y,x)$ is $1$ exactly when $y \notin N_{t_{j-1}}[x]$.
        \end{itemize}
    
        \item Otherwise, if $i \neq t_j$ (we are not at the beginning of the phase), let $x_{i-1}$ be the vertex removed from $\bG_{i-1}$ to form $\bG_i$. For each $v \in N_{i-1}(x_{i-1})$, we set 
        \[
            \delta_{i}(v) = \delta_{i-1}(v) - |N_{i-1}(v) \setminus N_{i-1}[x_{i-1}]|,
        \]
        where we compute 
        \[
            |N_{i-1}(v) \setminus N_{i-1}[x_{i-1}]| = |N_{t_j}(v) \setminus N_{t_j}[x_{i-1}]| - |(N_{t_j}(v) \setminus N_{t_j}[x_{i-1}]) \cap D|.
        \]
        Here, $D = V_{t_j} \setminus V_{i-1}$ is the set of vertices removed so far in the current phase. We compute the intersection by iterating over all nodes in $D$ in time $\Oh(|D|) = \Oh(\Delta)$.
    \end{enumerate}
    
    \subparagraph{Correctness.}
    The correctness relies on maintaining the exact number of missing edges in each vertex's neighborhood, and should be obvious to see.

    \subparagraph{Runtime.}
    Over the entire algorithm, we perform $\lceil n/\Delta \rceil$ phase initializations. The first initialization takes $\Oh(n^{\omega})$ time to compute $A_0 \bar{A}_0$. Each subsequent phase initialization computes the rectangular matrix product of an $n \times |D_j|$ matrix and a $|D_j| \times n$ matrix. Since $|D_j| \le \Delta$, this takes $\Oh(\mathsf{MM}(n, \Delta, n))$ time per phase, requiring a total time of $\Oh(n/\Delta \cdot \mathsf{MM}(n, \Delta, n))$. (Also note that $\Oh(n^{\omega}) = \Oh(n/\Delta \cdot \mathsf{MM}(n, \Delta, n))$ for every $\Delta$, so the latter term dominates the former.)

    The updates within the phases take $\Oh(\Delta)$ time for each neighbor of the removed vertex. Over all $n$ removals, the sum of the degrees of the removed vertices is bounded by $2m$. This yields a total intra-phase update time of $\Oh(m \cdot \Delta)$. 
\end{proof}

\subsection{Lower Bound for Finding $uv$-trails}

We show that verifying whether a graph $\bG$ contains a $uv$-trail with $d_{\bG}(u,v) \geq 3$, even for a single fixed pair of vertices $u, v$, cannot be performed faster than the time required to detect whether a graph contains a triangle. This is in stark contrast with the case $d_{\bG}(u,v)=2$, where we can solve the same problem in time $\Oh(n^2)$.
We start by proving the following for the case $d_{\bG}(u,v)=3$.

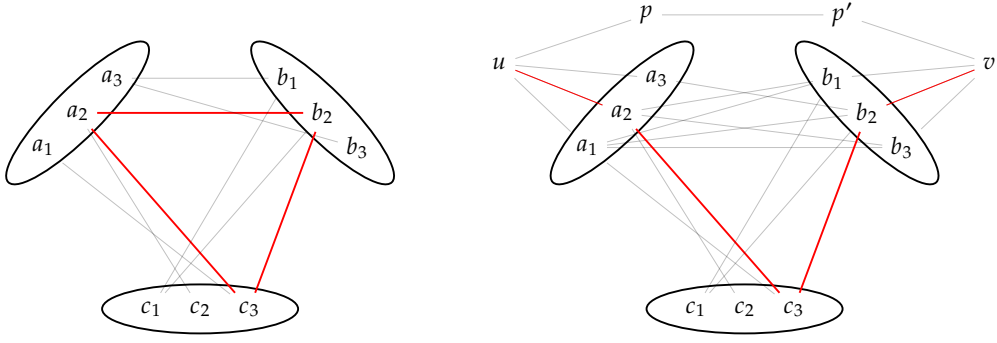
\begin{figure*}[t]
   \centering
   \scalebox{0.9}{\begin{tikzpicture}[scale=0.8]
    \begin{scope}[scale=0.9]
        \begin{scope}[shift={(-2.5,4)},rotate=45]
            \foreach \i in {1,...,3} {
              \node (L\i) at (1*\i-2.0,0) {$a_{\i}$};
            }
            \draw[thick] (0,0) ellipse (2.0 and 0.5);
        \end{scope}
        
        \begin{scope}[shift={(2.5,4)},rotate=-45]
            \foreach \i in {1,...,3} {
              \node (R\i) at (1*\i-2.0,0) {$b_{\i}$};
            }
            \draw[thick] (0,0) ellipse (2.0 and 0.5);
        \end{scope}
        
        \foreach \i in {1,...,3} {
          \node (B\i) at (1*\i-2.0,0) {$c_{\i}$};
        }
        \draw[thick] (0,0) ellipse (2.0 and 0.5);

        \draw[red, thick] (R2) -- (L2);
        \draw[opacity=0.25] (L3) -- (R1);
        \draw[opacity=0.25] (R3) -- (L3);

        \draw[opacity=0.25] (L1) -- (B3);
        \draw[opacity=0.25] (L2) -- (B2);
        \draw[red, thick] (L2) -- (B3);
        
        \draw[opacity=0.25] (R1) -- (B1);
        \draw[opacity=0.25] (R2) -- (B1);
        \draw[red, thick] (R2) -- (B3);
    \end{scope}
    
    \begin{scope}[shift={(10,0)},scale=0.9]
        \node (u) at (-5, 5) {$u$};
        \node (v) at (5, 5) {$v$};

        \node (P1) at (-2, 6) {$p$};
        \node (P2) at (2, 6) {$p'$};
        \draw[opacity=0.25] (u) -- (P1) -- (P2) -- (v);
    
        \begin{scope}[shift={(-2.5,4)},rotate=45]
            \foreach \i in {1,...,3} {
              \node (L\i) at (1*\i-2.0,0) {$a_{\i}$};
              \draw[opacity=0.25] (u) -- (L\i);
            }
            \draw[thick] (0,0) ellipse (2.0 and 0.5);
        \end{scope}
        
        \begin{scope}[shift={(2.5,4)},rotate=-45]
            \foreach \i in {1,...,3} {
              \node (R\i) at (1*\i-2.0,0) {$b_{\i}$};
              \draw[opacity=0.25] (R\i) -- (v);
            }
            \draw[thick] (0,0) ellipse (2.0 and 0.5);
        \end{scope}
        
        \foreach \i in {1,...,3} {
          \node (B\i) at (1*\i-2.0,0) {$c_{\i}$};
        }
        \draw[thick] (0,0) ellipse (2.0 and 0.5);

        % \draw[red, thick] (R2) -- (L2);
        % \draw[opacity=0.25] (R1) -- (L3);
        % \draw[opacity=0.25] (R3) -- (L3);

        \draw[opacity=0.25] (R1) -- (L1);
        \draw[opacity=0.25] (R1) -- (L2);
        \draw[opacity=0.25] (R2) -- (L1);
        \draw[opacity=0.25] (R2) -- (L3);
        \draw[opacity=0.25] (R3) -- (L1);
        \draw[opacity=0.25] (R3) -- (L2);

        \draw[opacity=0.25] (L1) -- (B3);
        \draw[opacity=0.25] (L2) -- (B2);
        \draw[red, thick] (L2) -- (B3);
        
        \draw[opacity=0.25] (R1) -- (B1);
        \draw[opacity=0.25] (R2) -- (B1);
        \draw[red, thick] (R2) -- (B3);

        \draw[red] (u) -- (L2);
        \draw[red] (R2) -- (v);
    \end{scope}
\end{tikzpicture}}
   \caption{
       An example of the graph $\bH$ used in \cref{lem:lb_uvtrail_gp}.
    }
   \label{fig:lb}
\end{figure*}

\begin{lemma}\label{lem:lb_uvtrail_gp}
    Let $\bG$ be an $n$-node graph. Then, in time $\Oh(n^2)$, we can build a $(3n+4)$-node graph $\bH$ containing two designated vertices $u$ and $v$ such that both of the following hold:
    \begin{enumerate}[(i)]
        \item $d_{\bH}(u,v) = 3$;
        \label{it:lb_uvtrail_gp:i}
        \item $\bH$ contains a $uv$-trail if and only if $\bG$ contains a triangle.
         \label{it:lb_uvtrail_gp:ii}
    \end{enumerate}
\end{lemma}
\begin{proof}
    Let $\bG = (V, E)$. First, we transform $\bG$ into a tripartite graph on $3n$ nodes by creating three disjoint copies of the vertex set, denoted $A, B$, and $C$, while retaining only those edges that connect vertices in different partitions. It is not difficult to see that the resulting tripartite graph contains a triangle if and only if the original graph $\bG$ contained a triangle.

    Therefore, assume that $\bG$ is a tripartite graph with vertex partitions $A = \{a_1, \ldots, a_n\}$, $B = \{b_1, \ldots, b_n\}$, and $C = \{c_1, \ldots, c_n\}$. We construct $\bH = (V', E')$ as follows (see Figure~\ref{fig:lb}):
    \begin{itemize}
        \item The vertex set is $V' = A \cup B \cup C \cup \{u, v, p, p'\}$, where $u, v, p$, and $p'$ are auxiliary vertices.
        \item The edge set $E'$ includes all original edges between $A$ and $C$ and between $C$ and $B$. However, it contains the complement of the edges between $A$ and $B$; that is, $E' \cap (A \times B) = \{\{a, b\} \in A \times B \mid \{a, b\} \notin E\}$.
        \item Additionally, $E'$ contains a path $u, p, p', v$ of length three connecting $u$ and $v$.
        \item Finally, $E'$ contains edges $\{u, a_i\}$ and $\{v, b_i\}$ for all $i \in \fragment{1}{n}$.
    \end{itemize}

    We observe that~\ref{it:lb_uvtrail_gp:i} holds. This is because the path $u, p, p', v$ has length three, and $d_{\bH}(u,v) \geq 3$ because $N[u] = A \cup \{u, p\}$ and $N[v] = B \cup \{v, p'\}$ are disjoint, so there is no path of length one or two.

    It remains to show that $\bH$ contains a $uv$-trail if and only if $\bG$ contains a triangle. First, assume $\bG$ contains a triangle on vertices $\{a_i, b_j, c_k\}$. We claim that $u, a_i, c_k, b_j, v$ is a $uv$-trail. To see this, notice that this path has length four and is chordless because $u$ and $v$ have no neighbors in $\{c_k, b_j\}$ and $\{a_i, c_k\}$, respectively, and $\{a_i, b_j\}$ is not an edge in $\bH$ because it was an edge in $\bG$.
    
    On the other hand, assume there exists a $uv$-trail $\pi$ in $\bH$. We observe that $\pi$ cannot utilize any nodes in $\{p, p'\}$. Furthermore, $\pi$ can contain at most one vertex from each set $A$ and $B$. If $\pi$ were to contain two distinct vertices $a_i, a_{i'}$ in that order, the edge $\{u, a_{i'}\}$ would constitute a chord; a similar argument holds for nodes in $B$. Thus, $\pi$ must follow the sequence: $u$, followed by some $a_i$, then vertices in $C$, then some $b_j$, and finally $v$. Since there are no edges between vertices within $C$, we have $\pi = u, a_i, c_k, b_j, v$ for some $c_k \in C$. Because $\pi$ is chordless and hence $\{a_i, b_j\} \notin E'$, it follows that $\bG$ contains the triangle $a_i, c_k, b_j$.
\end{proof}

This yields the following result.

\lbuvtrail

\begin{proof}
    For $k = 3$, we can use the construction from \cref{lem:lb_uvtrail_gp}: given an input graph for triangle detection, the graph $\bH$ is constructed and then the $uv$-trail detection algorithm from the problem statement is used.

    For $k > 3$, we note that the graph $\bH$ can still be constructed following \cref{lem:lb_uvtrail_gp}. In this case, however, we augment $\bH$ by adding a path of length $k-3$ connecting a new auxiliary node $u'$ to $u$. We then verify the existence of a $u'v$-trail. This construction remains valid because any trail from $u'$ to $v$ must necessarily pass through $u$, and the appended path segment is chordless.
\end{proof}

\section*{Acknowledgments}
OpenAI Codex was used to assist with revisions to the technical exposition, the layout of the figures, and the LaTeX formatting of this version. The authors are responsible for the final content.

\bibliographystyle{alphaurl}
\bibliography{main}

\newpage
\appendix

\section{Motivation for the Kratsch-Spinrad problems}\label{app:motivation}
Almost all the structures considered by Spinrad and Kratsch are fundamental to structural graph theory. They are closely intertwined and can be understood through four central graph classes:
\begin{enumerate}
    \item The first class consists of \emph{perfect graphs}, defined as graphs where every induced subgraph has a chromatic number equal to the size of its maximum clique. Perfect graphs are best known for the Strong Perfect Graph Theorem, conjectured by Berge in the 1960s \cite{B60, B61} and proven by Chudnovsky, Robertson, Seymour, and Thomas \cite{CNST06}. 
    In this context, \emph{star-cutsets} were first introduced by Chv\'atal \cite{C85} to characterize minimal imperfect graphs, establishing star-cutsets as a powerful structural obstruction used to decompose graph classes and prove perfection \cite{C85, R87}.

    \item \emph{Triangulated graphs} (also called \emph{chordal}, \emph{rigid-circuit}, \emph{monotone transitive}, or \emph{perfect elimination} graphs) are graphs in which every cycle of length four or more has a chord. Dirac~\cite{D61} gave two equivalent characterizations: a graph is chordal if and only if every induced subgraph is either a clique or contains two non-adjacent \emph{simplicial vertices} or, equivalently, if and only if every minimal cutset is a \emph{clique-cutset}. The latter formulation highlights why many problems can be solved faster on chordal graphs \cite{G72, RTS76}: a clique cutset decomposes a graph into independently solvable pieces, since the separating clique interacts uniformly with every part of the split. This property of clique cutsets can be exploited for other graph classes as well.

    \item A superclass of triangulated graphs is the class of \emph{weakly triangulated graphs}, in which \emph{two-pairs} play a fundamental role. Weakly triangulated graphs were first defined by Hayward \cite{H85} as graphs containing no induced \emph{holes} (that is, induced cycles of length at least five) or \emph{antiholes} (holes in the complement). Note that this is a generalization of triangulated graphs, since forbidding induced cycles of length four and five already forbids all antiholes. Sritharan conjectured that weakly triangulated graphs are exactly the graphs in which every induced subgraph is either a complete graph or contains a two-pair (so two-pairs play a role analogous to that of simplicial vertices). This was later proven by Chv\'atal and Rusu \cite{CR93} (see also \cite{ES95, CRS02}). Two-pairs are also central in the design of efficient algorithms for weakly triangulated graphs \cite{HTM90, HPS07}.

    \item Last but not least, \emph{\AT-free} graphs were first introduced by Lekkerkerker and Boland~\cite{LB62} as a way to relate interval graphs and triangulated graphs: they showed that a graph is an interval graph if and only if it is both chordal 
    and \AT-free. \AT-free graphs were subsequently studied in relation to other graph classes, such as cocomparability graphs and perfect graphs~\cite{G67, LB79, GMT84}, before being investigated in their own right in the seminal work of Corneil, Olariu, and Stewart~\cite{COS97}. One property of interval graphs that \cite{COS97} extended to \AT-free graphs is that every such graph contains a \emph{dominating pair} realizing the diameter. These and other properties of \AT-free graphs allow several $\mathsf{NP}$-hard problems to be solved in polynomial time on \AT-free graphs (e.g., \cite{S12, KM12, BKKM99, K99, KMT08}); moreover, recent algorithmic results on \AT-free graphs have even achieved running times that bypass the fine-grained lower bounds~\cite{D22, DDHV25}.
\end{enumerate}

\section{Dominated vertex is equivalent to OV}\label{sec:dom}
Recall that a vertex $u$ in a graph $\bG$ is dominated by a node $v$ if every neighbor of $u$ is also a neighbor of $v$ or is $v$. We will present two reductions. The first will reduce the problem of finding a dominated vertex to OV in $n$ dimensions. The second will reduce the OV problem to finding a dominated vertex in a given graph. We will then conclude with a short extra proof, that OV in $n$ dimensions is equivalent to Dominated vertex which is also equivalent to solving $\exists x\in V,\exists y\in N(v),\forall z\in N(v), (y,z)\in E$.

\begin{lemma}
Given an $n$-node graph $\bG=(V,E)$ one can create in $\Oh(n^2)$ time an instance of OV on two sets of $n$ vectors on $2n$ dimensions so that it has a pair of orthogonal vectors if and only if $G$ has a dominated vertex.
\end{lemma}

\begin{proof}
Let $A$ be the $n\times n$ adjacency matrix of $\bG$ and let $C$ be its complement. Let $I$ be the $n\times n$ identity matrix. Let $X$ be the $n\times 2n$ matrix obtained by concatenating $A$ and $I$ (with $I$ to the right of $A$).  Now let $Y$ be the $2n\times n$ matrix obtained by concatenating $C$ on top of $I$.

Note that for any given $i,j\in [n]$, we have that $(XY)[i,j]=0$ iff 
\begin{itemize}
\item for every $k$ such that $(i,k)\in E$, we have $(k,j)\in E$ (i.e. if $A[i,k]=1,$ then $C[k,j]=0$), and 
\item also $i\neq j$ (i.e. $I[i,i]\cdot I[i,j]=0$). \qedhere
\end{itemize}
\end{proof}

Note that the standard reduction from OV on two sets to OV on a single set still works here: you can create a single $2n\times (2n+2)$ matrix, $Z$, by concatenating $X$ on top of the complement of $Y$ and add two new columns, so that every row in $Z$ corresponding to a row of $X$ has $01$ in those two columns, and every row in $Z$ corresponding to a column of $Y$ has $10$ in those two columns.

Now we give the second direction of the reduction.

\begin{lemma}\label{lem:OVdomvertex}
Given two sets $U$ and $V$ of $n$ Boolean vectors on $n$ coordinates, in $\Oh(n^2)$ time one can create an $\Oh(n)$-node graph $G$ so that $G$ contains a dominated vertex if and only if there are $u\in U,v\in V$ s.t. $u\cdot v=0$.
\end{lemma}

\begin{proof}
Let $C$ be the $n$ coordinates of the two $n$-sets of $n$-length vectors $U$ and $V$.

Create a graph as follows (refer to \cref{fig:lOVdomvertex} for a visualization): Create vertex sets $X,Y,X',Y',X'',C',C''$ of size $n$ each. $C'$ and $C''$ are nodes corresponding to the coordinates $C$. They are connected by a matching: for every $c\in C$ there is an edge between its copies $c'\in C'$ and $c''\in C''$. Both $C'$ and $C''$ are cliques. 

$X$ is an independent set, $Y$ is a clique, and there is an edge between every node in $X$ and every node in $Y$. Every node $x\in X$ has a copy $x'\in X'$ and an edge to that copy; $X'$ is a clique. Similarly, every node $y\in Y$ has a copy $y'\in Y'$ and an edge to that copy; $Y'$ is a clique.
Every node in $X'$ has a copy in $X''$ and an edge to that copy and $X''$ is a clique.

Every node in $X'$ has an edge to every node in $Y$. Each $x_i\in X$ corresponds to a vector $u_i\in U$ and every $y_j\in Y$ corresponds to vector $v_j\in V$.
Each $x_i\in X$ has an edge to $c\in C'$ iff $u_i$ has a $1$ in position $c$. Each $y_j\in Y$ has an edge to $c\in C'$ iff $v_j$ has a $0$ in position $c$.

\begin{figure*}[t]
   \centering
   \scalebox{0.9}{\begin{tikzpicture}[
    font=\small,
    >=stealth,
    % --- Styles ---
    cliquebox/.style={
        draw, rounded corners=3pt,
        fill=blue!15, 
        minimum width=3.2cm,
        minimum height=1.4cm,
        align=center,
        thick
    },
    independentbox/.style={
        draw, rounded corners=3pt,
        fill=white,
        minimum width=3.2cm,
        minimum height=1.4cm,
        align=center,
        thick
    },
    matching/.style={black, line width=0.8pt},
    biclique/.style={black!70!white, line width=10pt, opacity=0.5, line cap=round, shorten <=2pt, shorten >=2pt},
    bicliquelabel/.style={font=\footnotesize, fill=white, fill opacity=0.9, text opacity=1, inner sep=2pt, rounded corners=2pt},
    incX/.style={blue!80!black, thick},
    incY/.style={red!80!black, thick},
    port/.style={circle, fill=black, inner sep=1.2pt},
    legend/.style={draw, rounded corners=3pt, align=left, font=\footnotesize, inner sep=6pt, fill=white, thick}
]

% --- 1. Nodes ---
% X Column (centered at x=0)
\node[independentbox] (X)   at (0,0)   {$X$\\[2pt]{\footnotesize independent}\\[1pt]{\scriptsize vectors $u_i\in U$}};
\node[cliquebox]      (Xp)  at (0,4.5) {$X'$\\[2pt]{\footnotesize clique}};
\node[cliquebox]      (Xpp) at (0,9.0) {$X''$\\[2pt]{\footnotesize clique}};

% Y Column (moved closer, centered at x=5.0)
\node[cliquebox]      (Y)   at (5.0, 1.5) {$Y$\\[2pt]{\footnotesize clique}\\[1pt]{\scriptsize vectors $v_j\in V$}};
\node[cliquebox]      (Yp)  at (5.0, 6.0) {$Y'$\\[2pt]{\footnotesize clique}};

% C Column (moved closer, centered at x=10.0)
\node[cliquebox]      (Cpp) at (10.0, 1.5) {$C''$\\[2pt]{\footnotesize coordinate clique}};
\node[cliquebox]      (Cp)  at (10.0, 6.0) {$C'$\\[2pt]{\footnotesize coordinate clique}};

% --- 2. Complete Bipartite Edges (Gray Bands) ---
\draw[biclique] (X.east) -- (Y.west) node[midway, bicliquelabel, yshift=10pt] {$X \times Y$ complete};
\draw[biclique] (Xp.east) -- (Y.west) node[midway, bicliquelabel, yshift=15pt] {$X' \times Y$ complete};

% --- 3. Matchings (Vertical Lines) ---
\foreach \offset in {-1.2, -0.8, -0.4, 0, 0.4, 0.8, 1.2}{
    % X to X'
    \draw[matching] ($(X.north)+(\offset,0)$) -- ($(Xp.south)+(\offset,0)$);
    % X' to X''
    \draw[matching] ($(Xp.north)+(\offset,0)$) -- ($(Xpp.south)+(\offset,0)$);
    % Y to Y'
    \draw[matching] ($(Y.north)+(\offset,0)$) -- ($(Yp.south)+(\offset,0)$);
    % C'' to C'
    \draw[matching] ($(Cpp.north)+(\offset,0)$) -- ($(Cp.south)+(\offset,0)$);
}

% Matching labels (offsets tightened for the compact layout)
\node[font=\footnotesize, fill=white, inner sep=2pt] at ($(X)!0.5!(Xp)   - (1.8,0)$) {$x_i x'_i$};
\node[font=\footnotesize, fill=white, inner sep=2pt] at ($(Xp)!0.5!(Xpp) - (1.8,0)$) {$x'_i x''_i$};
\node[font=\footnotesize, fill=white, inner sep=2pt] at ($(Y)!0.5!(Yp)   + (1.8,0)$) {$y_j y'_j$};
\node[font=\footnotesize, fill=white, inner sep=2pt] at ($(Cpp)!0.5!(Cp) + (1.8,0)$) {$c'_k c''_k$};

% --- 4. Incidence Edges (Blue and Red Curves) ---
\coordinate (X_out1) at ($(X.north east) + (-0.9, 0)$);
\coordinate (X_out2) at ($(X.north east) + (-0.3, 0)$);
\coordinate (X_out3) at ($(X.east)       + (0, 0.3)$);

\coordinate (Y_out1) at ($(Y.north east) + (-0.8, 0)$);
\coordinate (Y_out2) at ($(Y.north east) + (-0.2, 0)$);

\coordinate (Cp_in1) at ($(Cp.west)      + (0, -0.3)$);
\coordinate (Cp_in2) at ($(Cp.west)      + (0, 0.3)$);
\coordinate (Cp_in3) at ($(Cp.south west)+ (0.4, 0)$);
\coordinate (Cp_in4) at ($(Cp.south)     + (-0.8, 0)$);
\coordinate (Cp_in5) at ($(Cp.south)     + (0.3, 0)$);

% Draw blue edges (X to C')
\draw[incX] (X_out1) to[out=30, in=210] (Cp_in1);
\draw[incX] (X_out2) to[out=25, in=200] (Cp_in2);
\draw[incX] (X_out3) to[out=15, in=190] (Cp_in3);

% Draw red edges (Y to C')
\draw[incY] (Y_out1) to[out=50, in=230] (Cp_in4);
\draw[incY] (Y_out2) to[out=35, in=215] (Cp_in5);

% Draw ports (dots) at connection points
\foreach \p in {X_out1, X_out2, X_out3, Y_out1, Y_out2, Cp_in1, Cp_in2, Cp_in3, Cp_in4, Cp_in5} {
    \node[port] at (\p) {};
}

% --- 5. Legend ---
% Centered below the middle column (Y)
% \node[legend, anchor=north] at (5.0, -1.2) {
%     Filled boxes are cliques: $X'', X', Y, Y', C', C''$.\\
%     White boxes are independent sets: $X$.\\
%     Thick gray bands represent complete bipartite edge sets.\\
%     Parallel black edges represent the four matchings/copy maps.\\
%     Blue/red curves are sample coordinate-incidence edges into $C'$.\\[4pt]
%     \textcolor{blue!80!black}{\textbf{Blue:}} $x_i c'_k \in E \iff u_i[k]=1$.\\
%     \textcolor{red!80!black}{\textbf{Red:}} $y_j c'_k \in E \iff v_j[k]=0$.
% };

\end{tikzpicture}}
   \caption{
       A visualization of the graph constructed in \cref{lem:OVdomvertex}. Here: \\[1pt]
       \begin{itemize}
           \item Filled boxes are cliques: $X'', X', Y, Y', C', C''$.
           \item  White boxes are independent sets: $X$.
           \item Thick gray bands represent complete bipartite edge sets.
           \item Parallel black edges represent the four matchings/copy maps.
           \item Blue/red curves are sample coordinate-incidence edges into $C'$. More specifically:
           \begin{itemize}
               \item \textcolor{blue!80!black}{\textbf{Blue:}} $x_i c'_k \in E \iff u_i[k]=1$.
               \item \textcolor{red!80!black}{\textbf{Red:}} $y_j c'_k \in E \iff v_j[k]=0$.
           \end{itemize}
       \end{itemize}
    }
   \label{fig:lOVdomvertex}
\end{figure*}

Now, if $u_i$ and $v_j$ are orthogonal, then $x_i\in X$ is a dominated vertex, dominated by $y_j\in Y$: $(x_i,y_j)\in E$, and the remaining neighbors of $x_i$ are in $X'\cup C'$: $y_j$ is adjacent to all of $X'$ and since $u_i,v_j$ are orthogonal, $v_j$ has a 0 for every coordinate $c$ for which $u_i$ has a 1, so the neighbors of $x_i$ are also neighbors of $y_j$. In fact this is an iff.

By construction, if we have any other pair of nodes $(a,b)$ not in $X\times Y$, 
neither can dominate the due to the copies and cliques.

We will be explicit here. If $U$ or $V$ contains an all-$0$ vector, then the OV instance is a yes-instance and we can output a graph of two adjacent vertices, which has a dominated vertex. Hence we may assume that neither $U$ nor $V$ contains an all-$0$ vector.
\begin{enumerate}
\item A node $x\in X$ cannot dominate another $x_1\in X$ because they each have an edge to a different node in $X'$.
\item A node $x\in X$ cannot dominate a node in $X'\cup X''\cup Y'\cup C'\cup C''$ since these sets are cliques and $x$ is missing an edge to some node in each of these cliques.
\item A node $a\in X\cup X'\cup X''\cup C'\cup C''$ cannot dominate any $y\in Y$ because $y$ has a neighbor in $Y'$ and $a$ does not. (This also means $x\in X$ does not dominate any node.)
\item A node $y'\in Y'$ cannot dominate any node in $Y$ because $Y$ is a clique and $y'$ only has one edge to that clique.
\item A node $y_1\in Y$ cannot dominate another node $y_2\in Y$ since they have edges to different nodes in $Y'$. (This also means that the nodes in $Y$ are not dominated by any other nodes in the graph.)
\item A node $x'\in X'$ cannot dominate any node in $X''\cup Y'\cup C'\cup C''$ since these sets are cliques and $x'$ is missing an edge to some node in each of these cliques.

\item A node $x'\in X'$ cannot dominate any node in $X$ because these nodes have neighbors in $C'$ and $x'$ does not.
\item A node $x'\in X'$ cannot dominate another node in $X'$ because the two nodes have edges to different nodes in $X''$. (This means that the nodes in $X'$ cannot dominate any other node.)

\item Consider a node $x'\in X'$. Its neighbors are in $X'\cup X''\cup X\cup Y$ and it has neighbors in all these sets. The only other nodes that have neighbors in all these sets are other nodes $x'_1\in X'$. But from the above point, we get that no $x'\in X'$ is dominated by any other node.

\item Consider a node $x''\in X''$. It can only potentially be dominated by nodes in $X''\cup X'$, but we already know that the nodes in $X'$ do not dominate any other nodes. Another node $x''_1$ in $X''$ also cannot dominate $x''$ since $x''$ and $x''_1$ have different neighbors in $X'$ due to the matching. Hence the nodes of $X''$ are not dominated by any other node.

\item Similarly, a node $y'\in Y'$ can only be dominated by nodes in $Y'\cup Y$, but since $Y'$ is a clique and each $y\in Y$ has only one neighbor in $Y'$, $y$ cannot dominate $y'$. Also, no other node of $Y'$ can dominate $y'$ since they have distinct unique neighbors in $Y$. Thus the nodes of $Y'$ are not dominated.

\item Analogous to the above point, the nodes of $C''$ are not dominated by any other node.

\item The nodes of $C'$ are not dominated by nodes outside of $C'$ since $C'$ is a clique and only nodes of $C'$ have neighbors to all nodes of $C'$ since $V$ does not contain the all $0$s vector. However, two nodes in $C'$ also cannot dominate each other because of the matching to $C''$ and hence no node of $C'$ can be dominated.

\item Consider a node $x\in X$.  As $x$ has neighbors in all the sets $X',Y,C'$, any node that dominates it must have neighbors in each of these sets. Only the nodes of $X\cup Y$ have this property, and we have already established that $x$ cannot be dominated by another node of $X$. Thus the only possible nodes to be dominated are the nodes of $X$ and they can only be dominated by nodes of $Y$.
\end{enumerate}

Hence, $G$ has a dominated vertex iff there is some $x\in X$ dominated by some $y\in Y$ which holds iff for every $c\in C'$ for which $x$ has an edge to, $y$ also has an edge to. This is if and only if the vector of $U$ corresponding to $x$ is orthogonal to the vector of $V$ corresponding to $y$.
\end{proof}

In the above construction, the dominated vertex can only be dominated by one of its neighbors. This also means that dominated vertex can be reduced to the problem $\exists x\in V, \exists y\in N(x) \forall z\in N(x), (y,z)\in E$.

This latter problem can be reduced back to OV on $2n$ dimensions pretty easily: We are given $\bG=(V,E)$ in which we want to solve $\exists x\in V, \exists y\in N(v) \forall z\in N(v), (y,z)\in E$.  Let $A$ be the adjacency matrix and let $C$ be its complement. Let $I$ be the $n\times n$ identity matrix. Let $X$ be the $n\times 2n$ matrix obtained by concatenating $A$ and $C+I$ (with $C+I$ to the right of $A$).  Now let $Y$ be the $2n\times n$ matrix obtained by concatenating $C$ on top of $I$. Note that $(XY)[i,j]=0$ if and only if
\begin{enumerate}[(1)]
    \item for every $k$ such that $(i,k)\in E$, either $k=j$ or $(k,j)\in E$ (i.e. if $A[i,k]=1$, then $C[k,j]=0$), and
    \item $(C+I)[i,j]=0$, so $i\ne j$ and $(i,j)\in E$.
\end{enumerate}

\begin{theorem}
The following three problems are equivalent (using $\Oh(n^2)$ time reductions):
\begin{itemize}
\item Dominated Vertex in a $\Theta(n)$-node graph,
\item OV on $\Theta(n)$ vectors in $\Theta(n)$ dimensions,
\item The problem of deciding $\exists x\in V, \exists y\in N(v) \forall z\in N(v), (y,z)\in E$ in a $\Theta(n)$-node graph. \lipicsEnd
\end{itemize}
\end{theorem}

\end{document}